\documentclass[11pt,letter]{article}
\usepackage[round]{natbib} % for command citet and maybe other cite related problems. See http://en.wikibooks.org/wiki/LaTeX/Bibliography_Management#Natbib
\usepackage{amsmath} % command \nonumber does not work without this. This package is most likely useful for many other things.
\numberwithin{equation}{section} % related to amsmath. Allows a numbering of the equations by section.
\usepackage{graphicx} % for the \includegraphics command.
\usepackage{float} % Allows for a manual control of where exactly the figures should appear.
\usepackage{enumitem} % for more compact lists with \begin{itemize}[noitemsep, nolistsep] and also for list with roman numbers.
\usepackage[font=footnotesize]{caption} % Set the text of all captions to small.
\usepackage{algorithm}
\usepackage{algpseudocode}
\usepackage{multirow}
\usepackage{fancyhdr} % to control heads and feet in the document
\usepackage{dsfont} % in order to use \mathds{1}, la fonction indicatrice.

\usepackage[OT1]{fontenc}
\usepackage[utf8]{inputenc}
\usepackage{amsmath}
\usepackage{enumerate}
\usepackage{amsfonts}
\usepackage{color}
\usepackage{amssymb}
\usepackage{amsthm}
\usepackage{dsfont}
\usepackage{hyperref}
\usepackage[hypcap]{caption}
\usepackage{booktabs} %%% for the command \toprule at least
\usepackage{framed}
\usepackage{tabularx}

\newtheoremstyle{myprop}% name
  {6pt}%Space above
  {3pt}%Space below
  {\itshape}%Body font
  {}%Indent amount 1
  {\scshape}% Theorem head font
  {.}%Punctuation after theorem head
  {.5em}%Space after theorem head 2
  {}%Theorem head spec (can be left empty, meaning ‘normal’)
\theoremstyle{myprop}
\newtheorem{proposition}{Proposition}[section]

\newtheorem{corollary}[proposition]{Corollary}
\newtheorem{conjecture}[proposition]{Conjecture}

\theoremstyle{remark}

\newtheoremstyle{example}% name
  {6pt}%Space above
  {3pt}%Space below
  {\small\sffamily}%Body font
  {}%Indent amount 1
  {\small\sffamily\slshape}% Theorem head font
  {.}%Punctuation after theorem head
  {.5em}%Space after theorem head 2
  {}%Theorem head spec (can be left empty, meaning ‘normal’)
\theoremstyle{example}
\newtheorem{example}[proposition]{Example}

\renewcommand{\ge}{\geqslant}

\usepackage{amsmath}
\usepackage{amsfonts}
\usepackage{amssymb}
\usepackage{hyperref}
\usepackage{geometry}
\usepackage{amsthm}
\usepackage{subfig}
\usepackage{color}
\usepackage{fancyvrb}
\usepackage{enumitem} % for more compact lists with \begin{itemize}[noitemsep, nolistsep] and also for list with roman numbers

\newcommand{\bigqquad}{\qquad \qquad \qquad \qquad \qquad}

\usepackage{framed,color,verbatim}
\definecolor{shadecolor}{rgb}{.9, .9, .9}

\newenvironment{code}%
   {\footnotesize\snugshade\verbatim}%
   {\endverbatim\endsnugshade}

\begin{document}

\title{Extending one-factor copulas}
\author{
    Nathan Uyttendaele\footnote{Universit\'e catholique de Louvain, Institut de Statistique, Biostatistique et Sciences Actuarielles, Voie du Roman Pays 20, B-1348 Louvain-la-Neuve, Belgium.} \\
    \small na.uytten@gmail.com \\ \small {\bf corresponding author}
	\and
	Gildas Mazo$^*$ \\
	\small gildas.mazo@uclouvain.be
}
\date{\today}
\maketitle

\begin{abstract}
So far, one-factor copulas induce conditional independence with respect to a latent factor. In this paper, we extend one-factor copulas to conditionally dependent models. This is achieved through new representations which allow to build new parametric factor copulas with a varying conditional dependence structure. We discuss estimation and properties of these representations. In order to distinguish between conditionally independent and conditionally dependent factor copulas, we provide a novel statistical test which does not assume any parametric form for the conditional dependence structure. Illustrations of our framework are provided through examples, numerical experiments, as well as a real data analysis.
\end{abstract}


\section{Introduction}
\label{sec:introduction}
Factor copulas refer to copulas which can be expressed by means of unobserved variables, the factors. Most of the time, only one univariate (and continuous) factor $X_0$ is used, and thus one talks about one-factor copulas. They are a hot topic of research at the moment, being written about by \cite{krupskii2013factor}, \cite{joeBook2014}, \cite{krupskiiJoeFactor2015}, \cite{NPCOFTF} or \cite{oh2015modelling}. Nonetheless, the scope of current one-factor copulas is still limited when it comes to the construction of parametric models.

First, only factors with a uniform distribution are currently considered to the literature. Yet, in applications, the identification of a factor may implicitely assume estimating its distribution. 
Second, studying the factor's impact on the dependence structure is not allowed, either. Indeed, in current one-factor copulas, only conditional independence is allowed. That is, the variables $U_1,\dots,U_d$, whose joint distribution is a copula, are independent conditionally on the factor $X_0=x_0$. This means that, for all $u_1,\dots,u_d\in [0,1]$,
\begin{align*}
  P(U_1\leq u_1,\dots,U_d\leq u_d\vert X_0=x_0)=\prod_{i=1}^d P(U_i\leq u_i\vert X_0=x_0).
  \end{align*}
As a result, current one-factor copulas are written \citep{krupskii2013factor} as

\begin{align}\label{eq:current one-factor copulas}
C(u_1,\dots,u_d)&=\int_0^1 \prod_{i=1}^d P(U_i\leq u_i\vert X_0=x_0) \, dx_0\\
&=\int_0^1 \prod_{i=1}^d \frac{\partial C_i(u_i, x_0)}{\partial x_0} \, dx_0\notag \\
&=\int_0^1 \prod_{i=1}^d C_{i\vert 0}(u_i\vert x_0) \, dx_0,\notag
\end{align}
where $C_i$ is some bivariate copula and $X_0$ has a uniform distribution.

As \eqref{eq:current one-factor copulas} allows to see, the task of modeling only amounts to choose a parametric form for $P(U_i\leq u_i\vert X_0=x_0), i \in \{1, \ldots, d\}$.
What if the practitioner assumes that the dependence grows with the factor's value? Or, what if the dependence structure remains the same, but is not conditional independence?

This paper introduces new representations for one-factor copulas, allowing one to construct novel parametric models. These representations cover many models of the literature in a nontrivial way, as seen in Section \ref{modelsexploration}. Section \ref{propertiesofEOFCs} explores various properties of the new representations, while Section \ref{sec:methods} explores estimation and also proposes a novel test to assess whether conditional independence may hold or not, without assuming any parametric form for the dependence structure. Section \ref{sec:numericalExperiments} presents the numerical experiments used to illustrate the testing procedure as well as a real data analysis.

\section{Extended one-factor copulas and nested extended one-factor copulas \label{sec:spectralRepresentation}}

Consider the \emph{the law of total probability}, 
\begin{align}\label{eq:law of total probability}
  C(u_1,\dots,u_d)&=P(U_1\leq u_1,\dots,U_d\leq u_d)\notag\\
  &=\int
  P(U_1\leq u_1,\dots,U_d\leq u_d\vert X_0=x_0) f_{0}(x_0)\, dx_0,
\end{align}
where the integral is taken over the support of $X_0$ and $f_0$ denotes its density. Equation \eqref{eq:law of total probability} is the one from which the formula of current one-factor copulas originated, given by \eqref{eq:current one-factor copulas}. One can easily see that one-factor copulas are a reformulation of the law of total probability in which the factor $X_0$ is uniformly distributed on $[0,1]$ and the variables $U_1,\dots,U_d$ are assumed to be independent conditionally on the factor.

To extend one-factor copulas, in addition to letting the density of $X_0$, $f_{0}$, be unspecified, it is proposed to reconsider the decomposition of 
\begin{equation*}
P(U_1\leq u_1,\dots,U_d\leq u_d\vert X_0=x_0).
\end{equation*}

Given $X_0=x_0$, certainly the vector $(U_1,\dots,U_d)$ has a distribution function, but it is not, in general, a copula, because $U_i\vert X_0=x_0$ is not, in general, uniformly distributed. By Sklar's theorem \citep{sklar1959fonction,nelsen_introduction_2006}, 
\begin{equation*}
P(U_1\leq u_1,\dots,U_d\leq u_d\vert X_0=x_0)
\end{equation*}
can be decomposed as a copula and a set of marginal distributions, that is,

\begin{multline}\label{eq:conditional Sklar}
 P(U_1\leq u_1,\dots,U_d\leq u_d\vert X_0=x_0)= \\ C_{x_0}(P(U_1\leq u_1\vert X_0=x_0),\dots,P(U_d\leq u_d\vert X_0=x_0)).
\end{multline}

If we let $x_0$ vary, both the copula $C_{x_0}$ and the margins $P(U_i\leq u_i\vert X_0=x_0)$, $i=1,\dots,d$, will be, in fact, conditional distributions.

\begin{example}\label{ex:general Gaussian factor copula}
  Consider \eqref{eq:law of total probability} with $X_0$ following an exponential distribution, as 
  \begin{align}\label{eq:f_0 exponential}
    f_{0}(x_0)=e^{-x_0},\qquad x_0>0.
  \end{align}
  Moreover, in \eqref{eq:conditional Sklar}, assume that
  \begin{align*}
    P(U_i\leq u_i\vert X_0=x_0)
    =\int_0^{u_i}\frac{\Gamma(1+x_0)}{\Gamma(x_0)}
    (1-t)^{x_0-1}\, dt,
  \end{align*}
  where 
  \begin{align}\label{eq:gamma function}
    \Gamma(z)=\int_0^\infty t^{z-1}e^{-t}\,dt,
    \qquad z>0,
  \end{align}
  is the well known Gamma function.
  Finally, assume that the density of $C_{x_0}$, $c_{x_0}$ is
  \begin{align}\label{eq:Gaussian copula} 
    c_{x_0}(u_1,\dots,u_d)
    =(\text{det }R(x_0))^{-1/2}
    \exp\left[
      -\frac{1}{2}z^\top([R(x_0)]^{-1}-I)z\right],
    \end{align}
where $z=(z_1,\dots,z_d)$, $z_i$ is the quantile of order $u_i$ of the standard normal distribution, $I$ is the $d\times d$ identity matrix, and
\begin{align}\label{eq:exchangeable correlation matrix}
  R(x_0)=
  \begin{pmatrix}
    1& & & \\
     &\ddots&\beta(x_0)& \\
     &\beta(x_0)&\ddots& \\
     & & &1
  \end{pmatrix},
\end{align}
where
\begin{align*}
  \beta(x_0)=e^{-x_0}.
\end{align*}
\end{example}
In Example \ref{ex:general Gaussian factor copula}, for a fixed $x_0$, the copula $C_{x_0}$ is a multivariate Gaussian copula with an exchangeable correlation matrix with parameter $\beta(x_0)=e^{-x_0}$. Likewise, the distribution of $U_i$ given $X_0=x_0$ is a beta distribution with parameters $1$ and $x_0$. By Sklar's theorem, $P(U_i\leq u_i\vert X_0=x_0)$ and $C_{x_0}$ can be set independently.

\begin{example}\label{ex:general Gumbel factor copula}
  Consider \eqref{eq:law of total probability} with $X_0$ following a Pareto distribution, as
  \begin{align}\label{eq:f_0 pareto}
    f_{0}(x_0)=x_0^{-2},\qquad x_0>1.
  \end{align}
  Moreover, in \eqref{eq:conditional Sklar}, assume that 
  \begin{align*}
    C_{x_0}(u_1,\dots,u_d)=
    \exp\left[-
      \bigg( (-\log u_1)^{x_0}+\ldots+(-\log u_d)^{x_0} \bigg)^{1/x_0}
    \right].
  \end{align*}
\end{example}
In Example \ref{ex:general Gumbel factor copula}, for a fixed $x_0$, the copula $C_{x_0}$ is recognized to be a Gumbel-Hougaard copula with parameter $x_0$, see \citet[p. 153]{nelsen_introduction_2006}. The margins $P(U_i\leq u_i\vert X_0=x_0)$, $i=1,\dots,d$ were not specified.

While examples such as Example \ref{ex:general Gaussian factor copula} and Example \ref{ex:general Gumbel factor copula} could be multiplied endlessly, there is a representation, presented below, which permits to get them all, and build  general parametric one-factor copulas quite easily. 
So, in view of both the law of total probability \eqref{eq:law of total probability} and the ``conditional Sklar's theorem'' \eqref{eq:conditional Sklar}, every one-factor copula can be written as
\begin{align}\label{eq:representation 1}
  &C(u_1,\dots,u_d)\notag\\
  =&\int C_{x_0}[P(U_1\leq u_1\vert X_0=x_0),\dots,
  P(U_d\leq u_d\vert X_0=x_0)] f_{0}(x_0)\, dx_0,
\end{align}
where, as in Examples \ref{ex:general Gaussian factor copula} and \ref{ex:general Gumbel factor copula}, $C_{x_0}$ is to be understood as a collection, running over $x_0$, of well defined $d$-variate copulas and where the integral is taken over the support of $X_0$. In representation \eqref{eq:representation 1}, as well as in Example \ref{ex:general Gaussian factor copula} and Example \ref{ex:general Gumbel factor copula}, letting $x_0$ vary induces a collection of copulas $\{C_{x_0}\}$ which reflects the change in the dependence structure as the factor varies.
For instance, in the Example \ref{ex:general Gaussian factor copula}, $C_{x_0}\to \Pi$ (pointwise) as $x_0\to\infty$, where $\Pi$ denotes the independence copula, that is, $\Pi(u_1,\dots,u_d)=u_1\cdots u_d$ for all $u_1,\dots,u_d\in[0,1]$. 
On the other hand, if $x_0\to 0$, then $C_{x_0}\to M$, where $M$ denotes the Fr\'echet-Hoeffding bound for copulas, that is, $M$ represents the complete positive dependence structure, with $M(u_1,\dots,u_d)=\min(u_1,\dots,u_d)$ for all $u_1,\dots,u_d\in[0,1]$.
The opposite happens in Example \ref{ex:general Gumbel factor copula}. We have that $C_{x_0}\to M$ whenever $x_0\to 0$ and $C_{x_0}\to \Pi$ whenever $x_0\to 1$.

Representation \eqref{eq:representation 1} can be recast in terms of standard uniform variables only. Let $Q_0=F_0^{-1}$ be the inverse of the factor's distribution function $F_0$. By the change of variables $u_0=F_0(x_0)$ in \eqref{eq:representation 1}, we have
\begin{align}\label{eq:representation 2}
  &C(u_1,\dots,u_d)\notag\\
  =&\int C_{x_0}[ 
    P(U_1\leq u_1\vert U_0=F_0(x_0)),\dots,P(U_d\leq u_d\vert U_0=F_0(x_0))] f_{0}(x_0)\, dx_0 \notag\\
  =&\int_{0}^{1} C_{Q_0(u_0)}^\diamond[ 
    P(U_1\leq u_1\vert U_0=u_0),\dots,P(U_d\leq u_d\vert U_0=u_0)]\,du_0\notag \\
    =&\int_{0}^{1}  C_{Q_0(u_0)}^\diamond[
      C_{1\vert 0}(u_1\vert u_0),\dots,C_{d\vert 0}(u_d\vert u_0)]\, du_0 \notag \\
	    =&\int_{0}^{1}  C_{Q_0(u_0)}^\diamond\bigg[\frac{\partial C_{1}(u_1,u_0)}{\partial u_0}, \ldots, \frac{\partial C_{d}(u_d,u_0)}{\partial u_0}\bigg]\, du_0,
\end{align}
where $(U_i,U_0)\sim C_{i}$, $i=1,\dots,d$. 

Examples \ref{ex:general Gaussian factor copula} and \ref{ex:general Gumbel factor copula} can be recast in view of \eqref{eq:representation 2}.

\begin{example}[continuation of Example \ref{ex:general Gaussian factor copula}]\label{ex:Gaussian u_0}
  From \eqref{eq:f_0 exponential}, we have $Q_0(u_0)=-\log(1-u_0)$, hence, for a fixed $u_0\in(0,1)$, $C_{Q_0(u_0)}^\diamond$ is a multivariate Gaussian copula with  correlation matrix given by
  \begin{align*}
  R(u_0)=
  \begin{pmatrix}
    1& & & \\
     &\ddots&\beta(u_0)& \\
     &\beta(u_0)&\ddots& \\
     & & &1
  \end{pmatrix},\qquad
  \beta(u_0)=e^{-Q_0(u_0)}=1-u_0.
\end{align*}
Furthermore, since $P(U_i\leq u_i\vert U_0=u_0)=P(U_i\leq u_i\vert X_0=Q_0(u_0))$, we have
\begin{align*}
  C_{i\vert 0}(u_i\vert u_0)
  =\int_0^{u_i}\frac{\Gamma(1+Q_0(u_0))}{\Gamma(Q_0(u_0))}
    (1-t)^{Q_0(u_0)-1}\, dt,
\end{align*}
so that the underlying bivariate copula is
\begin{align*}
  C_{i}(u_i,u_0)=\int_0^{u_0}\int_0^{u_i}
  \frac{\Gamma(1+Q_0(t))}{\Gamma(Q_0(t))}
  (1-y)^{Q_0(t)-1}\,dt\, dy,
\end{align*}
for $i=1,\dots,d$.
\end{example}
In Example \ref{ex:Gaussian u_0}, note that $u_0=0$ implies $\beta(u_0)=1$, and thus $C_{Q_0(u_0)}^\diamond$ is the Fr\'echet-Hoeffding bound $M$. Likewise, $u_0=1$ implies $\beta(u_0)=0$ (by continuity), and thus $C_{Q_0(u_0)}^\diamond$ is the independence copula. 

\begin{example}[continuation of Example \ref{ex:general Gumbel factor copula}]\label{ex:Gumbel u_0}
    From \eqref{eq:f_0 pareto}, we have $Q_0(u_0)=1/(1-u_0)$, hence, for a fixed $u_0\in(0,1)$, 
    \begin{align*}
      &C_{Q_0(u_0)}^\diamond(u_1,\dots,u_d)
      =\exp\left[-
        \left( (-\log u_1)^{\beta(u_0)}+\ldots+(-\log u_d)^{\beta(u_0)} \right)^{1/\beta(u_0)}
      \right],\qquad \\ &\beta(u_0)=Q_0(u_0)=\frac{1}{1-u_0},
    \end{align*}
that is, $C_{Q_0(u_0)}^\diamond$ is a multivariate Gumbel-Hougaard copula with parameter given by 
\begin{equation*}
\beta(u_0)=Q_0(u_0)=1/(1-u_0).
\end{equation*}
\end{example}
In Example \ref{ex:Gumbel u_0}, $u_0=0$ implies $\beta(u_0)=1$, and thus $C_{Q_0(u_0)}^\diamond$ is the independence copula. Likewise, $u_0=1$ implies $\beta(u_0)=\infty$ (by continuity), and thus $C_{Q_0(u_0)}^\diamond$ is the Fr\'echet-Hoeffding bound. In short, we simply replaced the $x_0$'s of Example \ref{ex:general Gaussian factor copula} and Example \ref{ex:general Gumbel factor copula} by $Q_0(u_0)$.
% Also, note that the vectors  
% \begin{equation*}
% (X_1,\dots,X_d\vert X_0=x_0)
% \end{equation*}
% and 
% \begin{equation*}
% (U_1,\dots,U_d\vert X_0=x_0)
% \end{equation*}
% have the same copula $C_{x_0}$, while $(U_1,\dots,U_d\vert U_0=u_0)$ has copula $C_{Q_0(u_0)}$.

Mathematically, both representations \eqref{eq:representation 1} and \eqref{eq:representation 2} are of course equivalent. It is worth stressing that these representations are better not to be taken as plain mathematical results, but rather as a convenient way to generate new parametric one-factor copula models, as was shown in the above examples. The advantage of the representation in \eqref{eq:representation 2} is that it involves copulas only and allows an easy comparison with \eqref{eq:current one-factor copulas}. This last equation indeed corresponds to \eqref{eq:representation 2} with $C_{Q_0(u_0)}^\diamond=\Pi$, $x_0=u_0$ and $X_0=U_0$. Representation \eqref{eq:representation 1} is however more convenient when one adopts a point of view centered on the factor itself.

Both representations \eqref{eq:representation 1} and \eqref{eq:representation 2} contain two $d$-dimensional copulas: one, $C_{x_0}$ or $C_{Q_0(u_0)}^\diamond$, that's part of the integrand, and one, $C(u_1, \ldots, u_d)$, that's the result of the integral. In the rest of this paper, the latter will often be referred to as the \textit{outer} copula, while the former will often be referred to as the \textit{inner} copula. In the rest of this paper, the copulas $C_1, \ldots, C_d$ in \eqref{eq:representation 2} will be referred to as the bivariate copulas or \textit{linking} copulas.

The inner copula in both representations is usually assumed to be a copula with one parameter, this parameter being explicitly related to $u_0$ through $Q_0$ in representation \eqref{eq:representation 2}. However, many $d$-dimensional copula have more than one parameter. In such cases, each parameter can be assumed to be related to $u_0$ through a different mapping $Q_l$, $l \in \{1, \ldots, p\}$, where $p$ is the number of parameters of the related inner copula, and one can write:

\begin{equation} \nonumber
 C(u_1,\dots,u_d)
	    =\int_{0}^{1}  C_{\{Q_l(u_0)\}}^\diamond\bigg[\frac{\partial C_{1}(u_1,u_0)}{\partial u_0}, \ldots, \frac{\partial C_{d}(u_d,u_0)}{\partial u_0}\bigg]\, du_0.
\end{equation}
In general however, we will just write $C_{\{Q_l(u_0)\}}^\diamond$ as $C_{Q_0(u_0)}^\diamond$ or even $C_{u_0}^\diamond$.

Since any $d$-dimensional copula can be used in representations \eqref{eq:representation 1} and \eqref{eq:representation 2} as inner copula, one can also consider using an outer copula as inner copula, nesting our construction in itself.

Let us rewrite expression \eqref{eq:representation 2} as

\begin{equation} \label{preNEOFC}
 C(u_1,\dots,u_d)
  =\int_{0}^{1}  C_{t_1}^{\diamond_1}\bigg[H_{11}(u_1,t_1), \ldots, H_{d1}(u_d,t_1)\bigg]\, dt_1,
\end{equation}
that is, we rewrite $C_{Q_0(u_0)}^\diamond$ as $C_{t_1}^{\diamond_1}$, replace $\frac{\partial C_{i}(u_i,u_0)}{\partial u_0}$ by $H_{i1}(u_i, t_1)$ and $u_0$ by $t_1$.

Next, define $C_{t_1}^{\diamond_1}$ as

\begin{equation} \nonumber
C_{t_1}^{\diamond_1}(v_1,\dots,v_d)=\int_0^1 C_{t_2}^{\diamond_2}\big(H_{12}(v_1 , t_2), \ldots, H_{d2}(v_d , t_2)\big) dt_2,
\end{equation}
where $C_{t_2}^{\diamond_2}$ is some $d$-variate copula, the parameters of which being each related to $t_2$ through a different mapping, and where $H_{i2}(v_i, t_2)=\frac{\partial C_{i2}(v_i , t_2)}{\partial t_2}$, $\{C_{i2}\}$ being a set of bivariate copulas, with $i \in \{1, \ldots, d\}$.

The outer copula $C(u_1,\dots,u_d)$ in Equation \eqref{preNEOFC} is then
\begin{equation} \nonumber
\begin{split}
&C(u_1,\dots,u_d) \\
=&\int_0^1 \int_0^1 C_{t_2}^{\diamond_2}\big(H_{12}(H_{11}(u_1 , t_1) , t_2), \ldots, H_{d2}(H_{d1}(u_d , t_1) , t_2)\big) \,dt_2 \,dt_1.
\end{split}
\end{equation}

Note that in the above nesting scheme, the $d$-variate copula $C_{t_1}^{\diamond_1}$ actually does not depend on $t_1$. In general, we will always assume that only the last $d$-variate copula $C_{t_w}^{\diamond_w}$ in a nesting chain actually depends on its related variable $t_w$, in order to avoid an extra layer of complexity. While working with this last assumption, a nested extended one-factor copula (NEOFC), hereafter denoted as $C_w^\odot$, can be written as

\begin{equation} \label{NEOFC}
\begin{split}
&C_w^\odot(u_1,\dots,u_d) \\
&=\int_{[0,1]^w} C_{t_w}^{\diamond_w}\big(G_{1}(u_1 ; t_w, \ldots, t_1), \ldots, G_{d}(u_d ; t_w, \ldots, t_1)\big) \,dt_w\ldots dt_1,
\end{split}
\end{equation}
where $w, d\geq2$, $G_i(u_i ; t_w, \ldots, t_1)$ is

\begin{equation} \label{Gi}
G_i(u_i ; t_w, \ldots, t_1)=H_{iw}^{t_w} \circ \dots \circ H_{i1}^{t_1}(u_i),
\end{equation} 
and $H_{ij}^{t_j}(\bullet)\equiv H_{ij}(\bullet, t_j)$, with
\begin{equation} \label{Hijs}
\begin{split}
H_{i1}(\bullet, t_1)=&\frac{\partial C_{i1}(\bullet , t_1)}{\partial t_1}, \\
H_{i2}(\bullet, t_2)=&\frac{\partial C_{i2}(\bullet , t_2)}{\partial t_2}, \\
\vdots& \\
H_{iw}(\bullet, t_w)=&\frac{\partial C_{iw}(\bullet , t_w)}{\partial t_w}.
\end{split}
\end{equation}

As an example, if $w=3$, we have
\begin{equation} \nonumber
\begin{split}
G_i(u_i ; t_3, t_2, t_1)=&H_{i3}^{t_3} \circ H_{i2}^{t_2} \circ H_{i1}^{t_1}(u_i) \\
=&\frac{\partial C_{i3}(H_{i2}^{t_2} \circ H_{i1}^{t_1}(u_i) , t_3)}{\partial t_3} \\
=&\frac{\partial C_{i3}(\frac{\partial C_{i2}(H_{i1}^{t_1}(u_i) , t_2)}{\partial t_2} , t_3)}{\partial t_3} \\
=&\frac{\partial C_{i3}(\frac{\partial C_{i2}(\frac{\partial C_{i1}(u_i , t_1)}{\partial t_1} , t_2)}{\partial t_2} , t_3)}{\partial t_3}.
\end{split}
\end{equation}

The block of equations (\ref{Hijs}) reveals that there is a set of $w$ bivariate copulas underlying $G_i(u_i ; t_w, \ldots, t_1)$. Since $i \in \{1, \ldots, d\}$, this means that there is always a set of $d\times w$ bivariate copulas underlying any NEOFC. The set $\{C_{i1}\}$, $i \in \{1, \ldots, d\}$, is the first layer of linking copulas, while the set $\{C_{iw}\}$ corresponds to the last layer of linking copulas.

To summarize, a one-factor copula (OFC) can be written as
\begin{equation} \label{OFC:finaleq}
C(u_1,\dots,u_d)=\int_{0}^1 \prod_{i=1}^d \bigg(\frac{\partial C_i(u_i, u_0)}{\partial u_0}\bigg) \,du_0,
\end{equation}
an extended one-factor copula (EOFC) as
\begin{equation} \label{EOFC:finaleq}
C(u_1,\dots,u_d)=\int_{0}^1 C_{u_0}^{\diamond}\bigg(\frac{\partial C_1(u_1, u_0)}{\partial u_0}, \ldots, \frac{\partial C_d(u_d, u_0)}{\partial u_0}\bigg) \,du_0,
\end{equation}
and a nested extended one-factor copual (NEOFC) as
\begin{equation} \label{NEOFC:finaleq}
\begin{split}
&C_w^\odot(u_1,\dots,u_d) \\
&=\int_{[0,1]^w} C_{t_w}^{\diamond_w}\big(G_{1}(u_1 ; t_w, \ldots, t_1), \ldots, G_{d}(u_d ; t_w, \ldots, t_1)\big) \,dt_w\ldots dt_1.
\end{split}
\end{equation}

Note that EOFCs are a special case of NEOFCs: setting $w=1$ in the last expression above means falling back to extended one-factor copulas. Also note that, while the inner copula of a NEOFC is assumed to depend on only one factor, $T_w$, a nested extended one-factor copula actually encompasses more than one factor: from $T_1$ to $T_w$.

\section{Properties of EOFCs and NEOFCs} \label{propertiesofEOFCs}
In this section, various properties of EOFCs and NEOFCs are given. Identifiability issues are also addressed.

\subsection{On the inner copula} \label{ontheinnercopula}

In order to generate new parametric families of one-factor copulas using \eqref{eq:representation 1} or \eqref{eq:representation 2}, one can act on 3 components: the bivariate or linking copulas $C_{i}$, $i \in \{1,\dots,d\}$, the factor's distribution, represented either by its density $f_0$ or by its quantile function $Q_0$, and the set of multivariate copulas $\{C_{x_0}\}=\{C^\diamond_{Q_0(u_0)}\}$. Depending on the choice for the inner copula $C_{x_0}$, three different forms of factor copulas can be made. This is illustrated by the following example.

\begin{example}\label{ex:complete example}
    Let $X_0$ follow an exponential distribution with parameter $\lambda>0$,
  \begin{align*}
    f_{0}(x_0)=\lambda e^{-\lambda x_0},\qquad x_0>0.
  \end{align*}
  
  For $i \in \{1,\dots,d\}$, let $C_{i}$ be a Clayton copula so that
  \begin{align}\label{eq:formula for Clayton bivariate copula}
    C_{i}(u_i,u_0)=[(u_i^{-\alpha_i}+u_0^{-\alpha_i}-1)]^{-1/\alpha_i}
    \qquad \alpha_i\ge 0.
  \end{align}
 
 Finally, let $c_{x_0}$, the density of $C_{x_0}$, be as in \eqref{eq:Gaussian copula} where $R(x_0)$ is as in \eqref{eq:exchangeable correlation matrix} and where
\begin{align}\label{eq:correlation in complete example}
 \beta(x_0)=e^{-\beta_0 -\beta_1 x_0},\qquad \beta_0,\beta_1\ge 0.
\end{align}
\end{example}

In Example \ref{ex:complete example}, we have built a parametric model for one-factor copulas which allow for different features. First, the number of parameters, $d+3$, is linear in $d$, the dimension. Second, as will be seen in Section \ref{sec:methods}, assuming a parametric form for the factor's distribution, one can estimate the related parameter $\lambda$ by maximum pseudo-likelihood, thus learning valuable information about a variable which, by definition, is never directly observed. Finally, one can control the growth rate of the dependence structure, relative to the change of the factor's value. For instance, in \eqref{eq:correlation in complete example}, a decrease in $\beta_1$ leads to an increase in $\beta(x_0), \forall x_0$, $\beta(x_0)$ being the correlation parameter. If one set $\beta_1=0$, then $\beta(x_0)=\exp(-\beta_0)$, and thus the correlation parameter, hence the conditional copula $C_{x_0}$, does not depend on $x_0$ anymore: we call this \emph{conditional invariance}, not to be mistaken with conditional independence. This last feature happens when $\beta_0=\infty$, implying a correlation parameter $\beta(x_0)=0$.

The three different forms of factor copulas are now formalized.

\begin{description}
\item[Conditional independence.] Conditional independent one-factor copulas are those such that the inner copula is the independence copula. They correspond exactly to copulas of the form \eqref{eq:current one-factor copulas}, described in \cite{krupskii2013factor}, and their interpretation is such that, given the factor's value $X_0=x_0$, the variables $U_1,\dots,U_d$ are independent. Let us note that, even in this simple case, the obtained models are quite  reasonable and useful, as was demonstrated not only in \cite{krupskii2013factor} or \cite{oh2015modelling}, but also in view of the vast literature about conditional independent models, see \cite{skrondal2007latent}. In Section \ref{test}, we provide a novel procedure in order to test the assumption of conditional independence. 
\item[Conditional invariance.] Conditional invariant one-factor copulas are those such that, in \eqref{eq:representation 1} for instance, $C_{x_0}=C_{x_0'}$ for any $x_0$ and $x_0'$. That is, there is conditional dependence, but this conditional dependence remains the same regardless of the factor's value.
\item[Conditional noninvariance.] Conditional noninvariant one-factor copulas are those which are not conditionally invariant. Note that, \emph{a fortiori}, they are not conditionally independent either. Here, the conditional dependence structure is allowed to change with the factor's value. In Example \ref{ex:general Gaussian factor copula}, $\beta(x_0)\to 0$ as $x_0\to \infty$ and therefore $C_{x_0}\to \Pi$, the independence copula. On the opposite, $\beta(x_0)\to 1$ as $x_0\to 0$ and thus $C_{x_0}\to M$, the Fr\'echet-Hoeffding upper bound, characterizing complete positive dependence.
\end{description}

Natural parametric one-factor copulas can be built with the help of Kendall's tau and Spearman's rho. Recall that, given a bivariate copula $C$, Kendall's tau is a dependence coefficient in $[-1,1]$ defined by
\begin{align}\label{eq:Kendall's tau formula}
  \tau=4\int_{[0,1]^2}C(u,v)\,dC(u,v)-1.
\end{align}
A value of $\tau\approx 0$ hints at independence, and $\tau\approx -1$ (respectively $\tau\approx +1$) indicates negative (respectively positive) dependence. Example \ref{ex:Kendall's tau} illustrates the procedure.

\begin{example}\label{ex:Kendall's tau}
  Let $X_0$ follow a standard uniform distribution and let $C_{x_0}$ be as
  \begin{align*}   
    C_{x_0}(u_1,\dots,u_d)=
    (u_1^{-\tau^{-1}(x_0)}+\dots+u_d^{-\tau^{-1}(x_0)}-d+1)^{-1/\tau^{-1}(x_0)},
  \end{align*}
  where $\tau^{-1}$ is the inverse map of
  \begin{align}\label{eq:Kendall's tau of Clayton}
    \tau(\beta)=\frac{\beta}{\beta+2},\qquad \beta > 0.
  \end{align}
\end{example}

In Example \ref{ex:Kendall's tau}, for a fixed $x_0$, $C_{x_0}$ is recognized to be a Clayton copula with parameter $\tau^{-1}(x_0)=2x_0/(1-x_0)$ for $x_0\in(0,1)$. 
In general, the procedure works as follows. First, choose a parametric family of copulas, here the family of Clayton copulas
\begin{align}\label{eq:family of Clayton copulas}
  C_{\beta}(u_1,\dots,u_d)=
  (u_1^{-\beta}+\dots+u_d^{-\beta}-d+1)^{-1/\beta},\qquad \beta > 0.
\end{align}
Second, compute Kendall's tau (there is only one, since all pairs have the same distribution), which in this case is given by \eqref{eq:Kendall's tau of Clayton}. Third, choose the distribution of $X_0$ so that its support corresponds to the range of the map induced by \eqref{eq:Kendall's tau of Clayton}, here $(0,1)$. Fourth and last, replace $\beta$ by $\tau^{-1}(x_0)$ in \eqref{eq:family of Clayton copulas}.

The conditional dependence structure in Example \ref{ex:Kendall's tau} goes from conditional independence to conditional complete dependence. Indeed, when $x_0 \rightarrow 0$, $\beta(x_0) \rightarrow 0$ and $C_{\beta(x_0)} \rightarrow \Pi$. If $x_0\to 1$ instead, $\beta(x_0) \rightarrow \infty$ and $C_{\beta(x_0)} \rightarrow M$, the Fr\'echet-Hoeffding upper bound for copulas.
If one rather defines $\beta(x_0) = - \log(x_0)$, then $\beta(x_0) \rightarrow \infty$ when $x_0 \rightarrow 0$ and $C_{\beta(x_0)} \rightarrow M$. Hence, in one case the dependence increases with respect to the factor, while in the other case it decreases.

\subsection{Densities of EOFCs and NEOFCs}
The density of an EOFC, as defined by \eqref{preNEOFC} is straighforward to get:

\begin{equation} \label{densityEOFC}
\begin{split}
&c(u_1, \ldots, u_d)=\frac{\partial^d C(u_1,\dots,u_d)}{\partial u_1 \ldots \partial u_d} \\
=& \int_{0}^1 c_{t_1}^{\diamond_1}(H_{11}(u_1, t_1), \ldots, H_{d1}(u_d, t_1))\prod_{i=1}^d c_{i}(u_i , t_1) dt_1,
\end{split}
\end{equation}
which is the integral of the product of one $d$-variate density and $d$ bivariate densities, and where $c_{t_1}^{\diamond_1}$ is the density of $C_{t_1}^{\diamond_1}$ while $c_i$ is the density of $C_i$.

The density of a NEOFC, as defined by \eqref{NEOFC}, is however less straighforward to obtain and requires first to figure out the result of $\frac{\partial G_{i}}{\partial u_i}$, also refer to \eqref{Gi} and \eqref{Hijs}. Observe that (chain rule):
\begin{equation} \nonumber
\frac{\partial H_{ij}^{t_j} \circ H_{il}^{t_l}(u_i)}{\partial u_i}=\frac{\partial H_{ij}(H_{il}^{t_l}(u_i), t_j) }{\partial H_{il}^{t_l}(u_i)} \frac{\partial H_{il}^{t_l}(u_i)}{\partial u_i}.
\end{equation}

Using this last result, $\frac{\partial G_{i}}{\partial u_i}$ can now be written as

\begin{equation} \label{dGi}
\begin{split}
\frac{\partial G_{i}}{\partial u_i}=& \frac{\partial H_{iw}^{t_w} \circ \dots \circ H_{i1}^{t_1}(u_i)}{\partial u_i} \\
=& \frac{\partial H_{iw}(H_{i,w-1}^{t_{w-1}} \circ \dots \circ H_{i1}^{t_1}(u_i), t_w)}{\partial H_{i,w-1}^{t_{w-1}} \circ \dots \circ H_{i1}^{t_1}(u_i)} \frac{\partial H_{i,w-1}^{t_{w-1}} \circ \dots \circ H_{i1}^{t_1}(u_i)}{\partial u_i} \\
=& \prod_{k=w}^{3}\Bigg(\frac{\partial H_{ik}(H_{i,k-1}^{t_{k-1}} \circ \dots \circ H_{i1}^{t_1}(u_i), t_k)}{\partial H_{i,k-1}^{t_{k-1}} \circ \dots \circ H_{i1}^{t_1}(u_i)}\Bigg)\frac{\partial H_{i2}(H_{i1}^{t_1}(u_i), t_2)}{\partial H_{i1}^{t_1}(u_i)}\frac{\partial H_{i1}^{t_1}(u_i)}{\partial u_i}.
\end{split}
\end{equation}

The last expression in (\ref{dGi}) is the product of $(w-2)+1+1=w$ fractions. The last one is a bivariate copula density and the others each involve a bivariate density. Starting with the fraction on the right of the last expression in (\ref{dGi}), we indeed have that
\begin{equation} \nonumber
\frac{\partial H_{i1}^{t_1}(u_i)}{\partial u_i}= \frac{\partial H_{i1}(u_i, t_1)}{\partial u_i}=\frac{\partial^2 C_{i1}(u_i, t_1)}{\partial u_i \partial t_1}=c_{i1}(u_i, t_1).
\end{equation}

If we replace $H_{i,k-1}^{t_{k-1}} \circ \dots \circ H_{i1}^{t_1}(u_i)$ by $\bigcirc$ so that 
\begin{equation} \nonumber
\frac{\partial H_{ik}(H_{i,k-1}^{t_{k-1}} \circ \dots \circ H_{i1}^{t_1}(u_i), t_k)}{\partial H_{i,k-1}^{t_{k-1}} \circ \dots \circ H_{i1}^{t_1}(u_i)}=\frac{\partial H_{ik}(\bigcirc, t_k)}{\partial \bigcirc},
\end{equation}
we see that
\begin{equation}
\frac{\partial H_{ik}(\bigcirc, t_k)}{\partial \bigcirc}=\frac{\partial^2 C_{ik}(\bigcirc, t_k)}{\partial \bigcirc \partial t_k}=c_{ik}(\bigcirc, t_k).
\end{equation}

Therefore, we can write that
\begin{equation} \label{final_dGi}
\frac{\partial G_{i}}{\partial u_i}=\prod_{k=w}^{3}\Bigg(c_{ik}\bigg(H_{i,k-1}^{t_{k-1}} \circ \dots \circ H_{i1}^{t_1}(u_i), t_k\bigg)\Bigg)\times c_{i2}(H_{i1}^{t_1}(u_i), t_2) \times c_{i1}(u_i, t_1).
\end{equation}
  
As an example, we can expand (\ref{final_dGi}) for $w=3$:
\begin{equation}
\begin{split}
\frac{\partial G_{i}}{\partial u_i}=& c_{i3}(H_{i2}^{t_2}(H_{i1}^{t_1}(u_i)), t_3) c_{i2}(H_{i1}^{t_1}(u_i), t_2) c_{i1}(u_i, t_1) \\
=& c_{i3}(\frac{\partial C_{i2}(\frac{\partial C_{i1}(u_i, t_1)}{\partial t_1}  , t_2)}{\partial t_2}, t_3) c_{i2}(\frac{\partial C_{i1}(u_i, t_1)}{\partial t_1}, t_2) c_{i1}(u_i, t_1).
\end{split}
\end{equation}

The density of a NEOFC can now be written:
\begin{equation} \label{eq:density_EOFC}
\begin{split}
&c_w^\odot(u_1, \ldots, u_d)=\frac{\partial^d C_w^\odot(u_1,\dots,u_d)}{\partial u_1 \ldots \partial u_d} \\
=& \int_{[0,1]^w} c_{t_w}^{\diamond_w}\big(G_{1}(u_1 ; t_w, \ldots, t_1), \ldots, G_{d}(u_d ; t_w, \ldots, t_1)\big) \prod_{i=1}^d \frac{\partial G_{i}}{\partial u_i} dt_w\ldots dt_1 \\
=& \int_{[0,1]^w} c_{t_w}^{\diamond_w}\big(\{G_i\}\big) \\
\times& \prod_{i=1}^d \Bigg[\prod_{k=w}^{3}\Bigg(c_{ik}\bigg(H_{i,k-1}^{t_{k-1}} \circ \dots \circ H_{i1}^{t_1}(u_i), t_k\bigg)\Bigg)\times c_{i2}(H_{i1}^{t_1}(u_i), t_2) \times c_{i1}(u_i, t_1) \Bigg] dt_w\ldots dt_1,
\end{split}
\end{equation}
where $c_{t_w}^{\diamond_w}$ is the density of $C_{t_w}^{\diamond_w}$ and $\{G_i\}=\{G_{1}(u_1 ; t_w, \ldots, t_1),\ldots, G_{d}(u_d ; t_w, \ldots, t_1)\}$.

The density of a NEOFC, as displayed in the last expression of (\ref{eq:density_EOFC}), is thus made of two parts: a $d$-variate copula density on the left, $c_{t_w}^{\diamond_w}\big(\{G_i\}\big)$, and a product that involves $d\times w$ bivariate copula densities on the right of $c_{t_w}^{\diamond_w}\big(\{G_i\}\big)$.

\subsection{On the margins of EOFCs and NEOFCs}
The multivariate margins of an EOFC are still EOFCs. 

\begin{proof}
This can be easily seen by integrating $u_1$ out of Equation (\ref{densityEOFC}):
\begin{equation} \nonumber
\begin{split}
&\int_0^1 c(u_1, \ldots, u_d)\, du_1 \\
=& \int_{0}^1 \int_{0}^1 c_{t_1}^{\diamond_1}(H_{11}(u_1, t_1), \ldots, H_{d1}(u_d, t_1))\prod_{i=1}^d c_{i}(u_i , t_1) dt_1\,du_1 \\
=& \int_{0}^1 \int_{0}^1 c_{t_1}^{\diamond_1}(H_{11}(u_1, t_1), \ldots, H_{d1}(u_d, t_1))\prod_{i=1}^d \frac{\partial H_{i}(u_i , t_1)}{\partial u_i} dt_1\,du_1.
\end{split}
\end{equation}
The integral requires the following change of variable to be solved:
\begin{equation} \nonumber
\begin{split}
u =& H_{11}(u_1, t_1),\\
du =& \frac{\partial H_{11}(u_1, t_1)}{\partial u_1} du_1,
\end{split}
\end{equation}
leading to

\begin{equation} \nonumber
\int_{0}^1 \int_{0}^1 c_{t_1}^{\diamond_1}(u, H_{21}(u_2, t_1), \ldots, H_{d1}(u_d, t_1))\prod_{i=2}^d \frac{\partial H_{i1}(u_i , t_1)}{\partial u_i} dt_1\,du.
\end{equation}

After reorganizing this last expression, one gets:

\begin{equation} \nonumber
\begin{split}
&\int_0^1 c(u_1, \ldots, u_d)\, du_1 \\
=& \int_{0}^1 \prod_{i=2}^d \frac{\partial H_{i1}(u_i , t_1)}{\partial u_i} \int_{0}^1 c_{t_1}^{\diamond_1}(u, H_{21}(u_2, t_1), \ldots, H_{d1}(u_d, t_1)) \,du\,dt_1.
\end{split}
\end{equation}

The integral related to $du$ is the margin of $c_{t_1}^{\diamond_1}$ in $H_{21}, \ldots, H_{d1}$ when its first argument is integrated out. The final result is

\begin{equation} \nonumber
c^{[2:d]}(u_2, \ldots, u_d) = \int_{0}^1 c^{[2:d], {\diamond_1}}_{t_1}(H_{21}(u_2, t_1), \ldots, H_{d1}(u_d, t_1)) \prod_{i=2}^d \frac{\partial H_{i1}(u_i , t_1)}{\partial u_i} dt_1,
\end{equation}
which is the density of an EOFC. 
\end{proof}

In general, in order to get a given margin from an EOFC, one just needs to drop the variables that are no longer needed in (\ref{densityEOFC}).

Similarly to this last result, the multivariate margins of a NEOFC are also NEOFCs. 

\begin{proof}
The proof is similar to what was done in the case of EOFCs. Start from (\ref{eq:density_EOFC}) and integrate $u_1$ out:

\begin{equation}
\begin{split}
&\int_{0}^1 c_w^\odot(u_1, \ldots, u_d) du_1 \\
=& \int_{0}^1 \int_{[0,1]^w} c_{t_w}^{\diamond_w}\big(G_{1}(u_1 ; t_w, \ldots, t_1), \ldots, G_{d}(u_d ; t_w, \ldots, t_1)\big) \prod_{i=1}^d \frac{\partial G_{i}}{\partial u_i} dt_w\ldots dt_1\,du_1.
\end{split}
\end{equation}

Perform a change of variable: 
\begin{equation} \nonumber
\begin{split}
u&=G_{1}(u_1 ; t_w, \ldots, t_1), \\
du&=\frac{\partial G_{1}(u_1;t_w, \ldots, t_1)}{\partial u_1} du_1, \\
\end{split}
\end{equation}
leading to

\begin{equation} \nonumber
\int_{0}^1 \int_{[0,1]^w} c_{t_w}^{\diamond_w}\big(u, G_2, \ldots, G_{d}\big) \prod_{i=2}^d \frac{\partial G_{i}}{\partial u_i} dt_w\ldots dt_1\,du.
\end{equation}

Reorganize this last expression:

\begin{equation} \nonumber
\int_{[0,1]^w} \prod_{i=2}^d \frac{\partial G_{i}}{\partial u_i} \int_{0}^1 c_{t_w}^{\diamond_w}\big(u, G_2, \ldots, G_{d}\big)  du\,dt_w\ldots dt_1.
\end{equation}

The integral related to $du$ is just the $(d-1)$ margin of $c_{t_w}^{\diamond_w}$ in $G_2, \ldots, G_d$. One can now write:

\begin{equation}
\begin{split}
&\int_{0}^1 c_w^{\odot}(u_1, \ldots, u_d)\,du_1 = c_{w}^{[2:d], \odot}(u_2, \ldots, u_d) \\
=& \int_{[0,1]^w} c_{t_w}^{[2:d], \diamond_w}\big(G_2(u_2;t_w, \ldots, t_1), \ldots, G_{d}(u_d;t_w,\ldots, t_1)\big) \prod_{i=2}^d \frac{\partial G_{i}}{\partial u_i} dt_w\ldots dt_1,
\end{split}
\end{equation}
which is the expression of a NEOFC density.
\end{proof}

In general, the margins of a NEOFC are easy to get: just drop in expression (\ref{eq:density_EOFC}) the elements corresponding to the variables that are no longer of interest.

\subsection{The upper and lower Fréchet-Hoeffding bounds}

The upper Fréchet–Hoeffding bound is defined as $M(u_1, \ldots, u_d)=\min(u_1,\allowbreak \ldots, u_d)$. It is always a copula, and corresponds to comonotone random variables.

\begin{proposition} \label{upperF}
In \eqref{eq:representation 2}, let $C^\diamond_{Q_0(u_0)}=\Pi$ and let $C_i(u_i, u_0)=\min(u_i, u_0),\allowbreak \forall i \in \{1, \ldots, d\}$, where $\min(u_i, u_0)=M(u_i, u_0)$ is the upper Fréchet-Hoeffding bound on $(U_i, U_0)$. Then, $C(u_1, \ldots, u_d)=\min(u_1, \ldots, u_d)$.
\end{proposition}
\begin{proof} 
The outer copula is 
\begin{equation} \nonumber
C(u_1, \ldots, u_d)=\int_0^1 \frac{\partial \min(u_1, u_0)}{\partial u_0} \ldots \frac{\partial \min(u_d, u_0)}{\partial u_0} du_0,
\end{equation}
and since $\frac{\partial \min(u_i, u_0)}{\partial u_0}$ is 0 as soon as $u_0>u_i$, one can write:

\begin{equation} \nonumber
C(u_1, \ldots, u_d)=\int_0^{\min(u_1, \ldots, u_d)} \frac{\partial \min(u_1, u_0)}{\partial u_0} \ldots \frac{\partial \min(u_d, u_0)}{\partial u_0} du_0.
\end{equation}

Moreover, as long as $u_0<\min(u_1, \ldots, u_d)$, we have that $\frac{\partial \min(u_i, u_0)}{\partial u_0}=\frac{\partial u_0}{\partial u_0}=1, \forall i \in \{1, \ldots d\}$. Therefore, we have

\begin{equation} \nonumber
C(u_1, \ldots, u_d)=\int_0^{\min(u_1, \ldots, u_d)} 1 \ldots 1 \,du_0=\min(u_1, \ldots, u_d).
\end{equation}

\end{proof}
This result will turn out very useful to show that OFCs are not suited to model hierarchical dependence in Section \ref{modelsexploration}.

The 2-dimensional lower Fréchet-Hoeffding bound is usually written as $W(u_1,\allowbreak u_2)=\max(u_1+u_2-1, 0)$. Note that, in contrast to the upper Fréchet-Hoeffding bound, its generalization, $W(u_1, \ldots, u_d)=\max(d-1+\sum_{i=1}^d u_i, 0)$, is not a copula anymore.

\begin{proposition} \label{lowerF}
In \eqref{eq:representation 2}, let $d=2$, $C^\diamond_{Q_0(u_0)}=\Pi$, $C_1(u_1, u_0)=\max(u_1+u_0-1, 0)$ and $C_2(u_2, u_0)=\min(u_2, u_0)$. Then, $C(u_1, u_2)=\max(u_1+u_2-1, 0).$
\end{proposition}
\begin{proof} 
The outer copula is
\begin{equation} \nonumber
C(u_1, u_2)=\int_0^1 \frac{\partial \max(u_1+u_0-1, 0)}{\partial u_0}\frac{\partial \min(u_2, u_0)}{\partial u_0} du_0,
\end{equation}
and since the integrand is 0 as long as $u_2<u_0$, one can write that

\begin{equation} \nonumber
C(u_1, u_2)=\int_{0}^{u_2} \frac{\partial \max(u_1+u_0-1, 0)}{\partial u_0} \frac{\partial \min(u_2, u_0)}{\partial u_0} du_0.
\end{equation}

Since, $\forall u_0 \in [0, u_2]$, $\frac{\partial \min(u_2, u_0)}{\partial u_0}=1$, one can write

\begin{align} \nonumber
C(u_1, u_2)&=\int_{0}^{u_2} \frac{\partial \max(u_1+u_0-1, 0)}{\partial u_0} du_0 \\ \notag
&=\max(u_1+u_0-1, 0)\vert_{u_0=u_2}-\max(u_1+u_0-1, 0)\vert_{u_0=0} \\ \notag
&=\max(u_1+u_2-1, 0), \notag
\end{align}
and the proof is complete.

\end{proof}

\subsection{Two degenerate cases} \label{extractionsub}

From Equation \eqref{EOFC:finaleq} or Equation \eqref{NEOFC:finaleq}, one can end up with 
\begin{equation} \nonumber
C(u_1, \ldots, u_d)=E_{U_0}\big[ C_{U_0}^\diamond(u_1, \ldots, u_d)\big]
\end{equation}
or with 
\begin{equation} \nonumber
C_w^\odot(u_1, \ldots, u_d)=E_{T_w} \big[ C^{\diamond_w}_{T_w}(u_1, \ldots, u_d)\big].
\end{equation}

\begin{proposition}
In Equation \eqref{EOFC:finaleq}, let $C_{i}$, $\forall i \in \{1, \ldots, d\}$, be such that $C_{i}(u_i, u_0)=u_i u_0$. Then, $C(u_1, \ldots, u_d)=E_{U_0}\big[ C_{U_0}^\diamond(u_1, \ldots, u_d)\big]$.
\end{proposition}

\begin{proof}
If $C_{i}$, $\forall i \in \{1, \ldots, d\}$, is such that $C_{i}(u_i, u_0)=u_i u_0$, and since we have that $\frac{\partial (u_i u_0)}{\partial u_0}=u_i$, Equation \eqref{EOFC:finaleq} becomes
\begin{equation} \nonumber
C(u_1,\dots,u_d)=\int_{0}^1 C_{u_0}^{\diamond}(u_1, \ldots, u_d) \,du_0=E_{U_0}\big[C_{U_0}^\diamond(u_1, \ldots, u_d)\big].
\end{equation}
\end{proof}

The reasoning is the same for a NEOFC: one just needs to let all linking copulas, for all layers, be the independence copula. Note that in case of conditional invariance, we have that $C(u_1, \ldots, u_d)=C^\diamond(u_1, \ldots, u_d)$ and $C_w^\odot(u_1, \ldots, u_d)=C^{\diamond_w}(u_1, \ldots, u_d)$, that is, the outer copula is directly equal to the inner one. This result is important as it means that any $d$-dimensional copula in the literature is an EOFC where the linking copulas are the independence copula and the inner copula is equal to the $d$-dimensional copula of interest.

A second degenerate case, regarding the linking copulas, is now presented.

\begin{proposition} \label{extraction1}
Let the inner copula be the independence copula and set all linking copulas to the upper Fréchet-Hoeffding bound, except for one, $C_l$. In $C(u_1, \ldots, u_d)$, set now all arguments to 1, except for $u_l$ and $u_k$, where $k$ is an arbitrary element of $\{1, \ldots d\}\backslash \{l\}$. The resulting bivariate margin of $C$, $C^{[l,k]}(u_l, u_k)$, is then equal to $C_l(u_l, u_k)$.
\end{proposition}

\begin{proof}
Let us write the expression of $C^{[l,k]}(u_l, u_k)$ when the inner copula is the independence one and $C_k$ is the upper Fréchet-Hoeffding bound:

\begin{equation} \nonumber
C^{[l,k]}(u_l, u_k)=\int_{0}^1 \frac{\partial C_l(u_l, u_0)}{\partial u_0} \frac{\partial \min(u_k, u_0)}{\partial u_0} du_0.
\end{equation}
The function $\min(u_k, u_0)$ will return $u_0$ as long as $u_0<u_k$. Otherwise, $u_k$ is returned. Therefore, one can write

\begin{align} \nonumber
C^{[l,k]}(u_l, u_k)&=\int_{0}^{u_k} \frac{\partial C_l(u_l, u_0)}{\partial u_0} \frac{\partial u_0}{\partial u_0} du_0 + \int_{u_k}^1 \frac{\partial C_l(u_l, u_0)}{\partial u_0} \frac{\partial u_k}{\partial u_0} du_0 \\ \notag
&=\int_{0}^{u_k} \frac{\partial C_l(u_l, u_0)}{\partial u_0} du_0 + 0 \\ \notag
&=C_l(u_l, u_k)-C_l(u_l, 0) \\ \notag
&=C_l(u_l, u_k) \\ \notag
\end{align}
\end{proof}

Proposition \ref{extraction1} can be generalized for the case of a NEOFC as follows: if $C_{lm}$ is the linking copula of interest, $l \in \{1, \ldots d\}$, $m \in \{1, \ldots w\}$, then the $[l,k]$-margin of $C_w^\odot(u_1, \ldots, u_d)$ is $C_{lm}(u_l, u_k)$ if $C_{t_w}^{\diamond_t}$ is the independence copula while the linking copulas related to $T_m$ are set to the upper Fréchet-Hoeffding bound, except for $C_{lm}$, and the remaining linking copulas are set to the independence copula. The proof that the result of this setup is $C^{[l,k], \odot}(u_l, u_k)=C_{lm}(u_l, u_k)$ is trivial in view of the proof of Proposition \ref{extraction1} and is therefore left to the reader.

\subsection{Dependence properties}
To start, let us write representation \eqref{eq:representation 2} when $d=2$:

\begin{equation}\label{EOFC2d}
C(u_1,u_2)=\int_{0}^{1}  C_{Q_0(u_0)}^\diamond \bigg[\frac{\partial C_{1}(u_1,u_0)}{\partial u_0}, \frac{\partial C_{2}(u_2,u_0)}{\partial u_0}\bigg]\, du_0.
\end{equation}

Let's also recall some dependence properties from \cite{nelsen_introduction_2006}.

A copula $C$ is positively quadrant dependent (PQD) if $C(\mathbf{u}) \geq \Pi(\mathbf{u})$ for any $\mathbf{u} \in [0,1]^d$, where $\Pi(\mathbf{u})$ is the independence copula. Similarly, a copula $C$ is negatively quadrant dependent (NQD) if $C(\mathbf{u}) \leq \Pi(\mathbf{u})$ for any $\mathbf{u} \in [0,1]^d$.

If $C_1$ is the copula of $(U_1, U_0)$, then $U_1$ is said to be stochastically decreasing in $U_0$ if $\frac{\partial C_1(u_1, u_0)}{\partial u_0}$ is an increasing function of $u_0$, $\frac{\partial^2 C_1(u_1, u_0)}{\partial u_0^2} \geq 0$. This can be written as $SD(U_1 \vert U_0)$ and it will often be said that $C_1$ is SD in $U_0$ or that $C_1(u_1, u_0)$ is SD in its second argument.

Similarly, if $C_1$ is the copula of $(U_1, U_0)$, then $U_1$ is said to be stochastically increasing in $U_0$ if $\frac{\partial C_1(u_1, u_0)}{\partial u_0}$ is a decreasing function of $u_0$, $\frac{\partial^2 C_1(u_1, u_0)}{\partial u_0^2} \leq 0$. This can be written as $SI(U_1 \vert U_0)$ and it will often be said that $C_1$ is SI in $U_0$ or that $C_1(u_1, u_0)$ is SI in its second argument.

Finally, note that if a copula $C_1$ is stochastically increasing in either $U_0$ or $U_1$, then it is PQD. The same applies if $C_1$ is stochastically decreasing in either $U_0$ or $U_1$: $C_1$ is then NQD. The reverse is however not true: being SI or SD for a copula is a stronger dependence property than being PQD or NQD.

\begin{proposition} \label{firstdep}
Let the inner copula in \eqref{EOFC2d}, $C_{Q_0(u_0)}^\diamond$, be a PQD copula, $\forall u_0$. Then if $C_2$ is NQD and $C_1$ is $SD(U_1\vert U_0)$ or if $C_2$ is PQD and $C_1$ is $SI(U_1\vert U_0)$, the outer copula is PQD, too.
\end{proposition}

\begin{proof}
If $C_{Q_0(u_0)}^\diamond(v_1, v_2)$ is PQD, then $C_{Q_0(u_0)}^\diamond(v_1, v_2) \geq v_1 v_2$. Therefore, it is clear that
\begin{multline}\nonumber
C(u_1,u_2)=\int_{0}^{1}  C_{Q_0(u_0)}^\diamond\bigg[\frac{\partial C_{1}(u_1,u_0)}{\partial u_0}, \frac{\partial C_{2}(u_2,u_0)}{\partial u_0}\bigg]\, du_0 \\ \geq\int_{0}^{1} \frac{\partial C_{1}(u_1,u_0)}{\partial u_0} \frac{\partial C_{2}(u_2,u_0)}{\partial u_0}\, du_0.
\end{multline}
Using integration by parts, one can then show that

\begin{equation}\nonumber
C(u_1,u_2)\geq u_2 \frac{\partial C_1(u_1, u_0)}{\partial u_0}\biggr\rvert_{u_0=1}-\int_{0}^{1} C_2(u_2, u_0) \frac{\partial^2 C_1(u_1, u_0)}{\partial u_0^2} du_0.
\end{equation}
Since $C_1$ is $SD(U_1\vert U_0)$ and $C_2$ is NQD, we have that
\begin{equation}\nonumber
0 \leq \int_{0}^{1} C_2(u_2, u_0) \frac{\partial^2 C_1(u_1, u_0)}{\partial u_0^2} du_0 \leq \int_{0}^{1} u_2 u_0 \frac{\partial^2 C_1(u_1, u_0)}{\partial u_0^2} du_0,
\end{equation}
and therefore

\begin{equation}\nonumber
C(u_1,u_2)\geq u_2 \frac{\partial C_1(u_1, u_0)}{\partial u_0}\biggr\rvert_{u_0=1}-\int_{0}^{1} u_2 u_0 \frac{\partial^2 C_1(u_1, u_0)}{\partial u_0^2} du_0.
\end{equation}
Using integration by parts again, we have that
\begin{equation}\nonumber
\int_{0}^{1} u_2 u_0 \frac{\partial^2 C_1(u_1, u_0)}{\partial u_0^2} du_0 = u_2 \frac{\partial C_1(u_1, u_0)}{\partial u_0}\biggr\rvert_{u_0=1}- u_1 u_2,
\end{equation}
leading to 
\begin{equation}\nonumber
C(u_1,u_2) \geq u_1 u_2.
\end{equation}
\end{proof}

A similar proposition as the one above can be made, ensuring that the outer copula is this time NQD. The proof is almost the same as the one from the previous proposition and is therefore left to the reader.

\begin{proposition} \label{ddweq}
Let the inner copula in \eqref{EOFC2d}, $C_{Q_0(u_0)}^\diamond$, be a NQD copula, $\forall u_0$. Then if $C_2$ is NQD and $C_1$ is $SI(U_1\vert U_0)$ or if $C_2$ is PQD and $C_1$ is $SD(U_1\vert U_0)$, the outer copula is NQD, too.
\end{proposition}

\begin{corollary}
Let the inner copula in a 2-dimensional EOFC be the independence copula, that is, we fall back to Equation \eqref{eq:current one-factor copulas}, where $X_0=U_0$. Then if $C_2$ is PQD, $C_1$ is $SI(U_1\vert U_0)$ and since the independence copula is PQD, by Proposition \ref{firstdep}, the outer copula is PQD.
Similarly, if $C_2$ is NQD, $C_1$ is $SI(U_1\vert U_0)$ and since the independence copula is NQD, by Proposition \ref{ddweq}, the outer copula is NQD.
\end{corollary}

Note that Proposition 1 in \cite{krupskii2013factor} offers a similar result as the one from the above corollary, where it is shown that if both $C_1$ and $C_2$ are SI in $U_0$, the outer copula is PQD. Proposition \ref{firstdep} shows that both $C_1$ and $C_2$ do not need to be SI in $U_0$ for the outer copula to be PQD: as long as one of the two linking copulas is SI and the remaining one is just PQD, then the outer one is PQD.

% \begin{proof}
% If $C_{Q_0(u_0)}(v_1, v_2)$ is NQD, then $C_{Q_0(u_0)} \leq v_1 \times v_2$. Therefore, it is clear that
% \begin{equation}\nonumber
% C(u_1,u_2)=\int_{0}^{1}  C_{Q_0(u_0)}\bigg[\frac{\partial C_{1}(u_1,u_0)}{\partial u_0}, \frac{C_{2}(u_2,u_0)}{\partial u_0}\bigg]\, du_0\leq\int_{0}^{1} \frac{\partial C_{1}(u_1,u_0)}{\partial u_0} \times \frac{C_{2}(u_2,u_0)}{\partial u_0}\, du_0.
% \end{equation}
% Using integration by parts, one can show that

% \begin{equation}\nonumber
% C(u_1,u_2)\leq u_2 \frac{\partial C_1(u_1, u_0)}{\partial u_0}\biggr\rvert_{u_0=1}-\int_{0}^{1} C_2(u_2, u_0) \frac{\partial^2 C_1(u_1, u_0)}{\partial u_0^2} du_0.
% \end{equation}
% Since $C_1$ is SI($U_1\vert U_0$) and $C_2$ is NQD, we have that
% \begin{equation}\nonumber
% 0 \geq \int_{0}^{1} C_2(u_2, u_0) \frac{\partial^2 C_1(u_1, u_0)}{\partial u_0^2} du_0 \geq \int_{0}^{1} u_2 u_0 \frac{\partial^2 C_1(u_1, u_0)}{\partial u_0^2} du_0,
% \end{equation}
% and therefore

% \begin{equation}\nonumber
% C(u_1,u_2)\leq u_2 \frac{\partial C_1(u_1, u_0)}{\partial u_0}\biggr\rvert_{u_0=1}-\int_{0}^{1} u_2 u_0 \frac{\partial^2 C_1(u_1, u_0)}{\partial u_0^2} du_0.
% \end{equation}
% Using integration by part again, we have that
% \begin{equation}\nonumber
% \int_{0}^{1} u_2 u_0 \frac{\partial^2 C_1(u_1, u_0)}{\partial u_0^2} du_0 = u_2 \frac{\partial C_1(u_1, u_0)}{\partial u_0}\biggr\rvert_{u_0=1}- u_1 u_2,
% \end{equation}
% leading to 
% \begin{equation}\nonumber
% C(u_1,u_2) \leq u_1 \times u_2.
% \end{equation}
% \end{proof}

Let us now write the expression of a NEOFC when $w, d=2$. Then:

\begin{equation}\label{NEOFC2d}
C_{2}^\odot (u_1,u_2)=\int_{[0,1]^2} C_{t_2}^{\diamond_2} \bigg(\frac{\partial C_{12}(\frac{\partial C_{11}(u_1 , t_1)}{\partial t_1} , t_2)}{\partial t_2}, \frac{\partial C_{22}(\frac{\partial C_{21}(u_2 , t_1)}{\partial t_1} , t_2)}{\partial t_2}\bigg) \,dt_2\, dt_1.
\end{equation}

The following proposition ensures that the outer copula is going to be PQD. 

\begin{proposition} \label{NEOFCpropungen}
Let the inner copula in \eqref{NEOFC2d}, $C_{t_2}^{\diamond_2}$, be a PQD copula, $\forall t_2$. Then if $C_{21}$ is NQD and $C_{11}$ is SD in $T_1$, while $C_{22}$ is NQD and $C_{12}$ is SD in $T_2$, the outer copula, $C_{2}^\odot (u_1,u_2)$, is PQD.
\end{proposition}

\begin{proof}
Let us write $C_2^\odot (u_1,u_2)$ as

\begin{equation}
C_{2}^\odot (u_1,u_2)=\int_{0}^1 C_{t_1}^{\diamond_1}\bigg(\frac{\partial C_{11}(u_1 , t_1)}{\partial t_1}, \frac{\partial C_{21}(u_2 , t_1)}{\partial t_1}\bigg) \,dt_1,
\end{equation}
where 
\begin{equation}
C_{t_1}^{\diamond_1}(v_1,v_2)=\int_{0}^1 C_{t_2}^{\diamond_2}\bigg(\frac{\partial C_{12}(v_1 , t_2)}{\partial t_2}, \frac{\partial C_{22}(v_2 , t_2)}{\partial t_2}\bigg) \,dt_2.
\end{equation}

By Proposition \ref{firstdep}, if $C_{t_1}^{\diamond_1}$ is PQD while $C_{21}$ is NQD and $C_{11}$ is SD in $T_1$, then $C_{2}^\odot (u_1,u_2)$ is PQD, too. By Proposition \ref{firstdep} again, in order for $C_{t_1}^{\diamond_1}$ to be PQD, one needs $C_{t_2}^{\diamond_2}$ to be PQD, while $C_{22}$ is NQD and $C_{12}$ is SD in $T_2$.
\end{proof}

Proposition \ref{NEOFCpropungen} can be generalized as follows:
\begin{proposition} \label{NEOFCpropgen}
Let the inner copula in a 2-dimensional NEOFC be PQD. If for each $j \in \{1, \ldots, w\}$ we have that $C_{1j}$ is NQD and $C_{2j}$ is SD with respect to the factor $T_j$ or that $C_{2j}$ is NQD and $C_{1j}$ is SD with respect to the factor $T_j$, then the outer copula, $C_w^\odot$, is PQD.
\end{proposition}

The proof of this last proposition is trivial in regard of the proof of Proposition \ref{NEOFCpropungen}. Proposition \ref{NEOFCpropgen} can also be expressed so that the resulting outer copula is NQD:

\begin{proposition}
Let the inner copula in a 2-dimensional NEOFC be NQD. If for each $j \in \{1, \ldots, w\}$ we have that $C_{1j}$ is NQD and $C_{2j}$ is SI with respect to the factor $T_j$ or that $C_{2j}$ is NQD and $C_{1j}$ is SI with respect to the factor $T_j$, then the outer copula, $C_w^\odot$, is NQD.
\end{proposition}

Hereafter is an example of NEOFC meeting the requirements of Proposition \ref{NEOFCpropgen}.
\begin{example} Start from \ref{NEOFC2d}. Replace the inner copula $C_{t_2}^{\diamond_2}$ by the Frank copula, with parameter
\begin{equation}\nonumber
\theta^{\diamond_2}(t_2)=\frac{\log(1-t_2)}{-\lambda},
\end{equation} 
that is, we use the inverse CDF of an exponential distribution on $t_2$ in order to calculate $\theta^{\diamond_2}$, where $\lambda > 0$ is the rate, that is, the maximum value of the corresponding exponential probability density function. Since the support of the exponential is $[0, \infty)$, the parameter $\theta^{\diamond_2}(t_2)$ is always positive and therefore $C_{t_2}^{\diamond_2}$ is PQD $\forall t_2 \in [0,1]$. 

Next, set both $C_{11}$ and $C_{12}$ to be a Mardia copula with a negative parameter's value, ensuring the copula is NQD \citep[p. 188]{nelsen_introduction_2006}, while both $C_{21}$ and $C_{22}$ are AMH copulas, each with a negative parameter's value, to ensure that $C_{21}$ is SD in $T_1$ and $C_{22}$ is SD in $T_2$. Indeed, for an AMH copula $C_{\theta_\text{AMH}}(\bullet,t_j)$, one have that
\begin{equation} \nonumber
\frac{\partial^2 C_{\theta_\text{AMH}}(\bullet,t_j)}{\partial t_j^2}=-\frac{2 \theta_\text{AMH}\times (\bullet-1)\bullet(\theta_\text{AMH}\times(\bullet-1)+1)}{(\theta_\text{AMH}\times(t_j-1)(\bullet-1)-1)^3},
\end{equation}
which is a positive quantity as long as $\theta_\text{AMH}<0$.

The outer copula resulting from this construction is, by Proposition \ref{NEOFCpropgen}, PQD.
\end{example}

Next we introduce an important conjecture, which describes the relationship between the inner copula and the outer copula of an EOFC.

\begin{conjecture} \label{taylorup}
Let the inner copula in \eqref{EOFC2d}, $C_{Q_0(u_0)}^\diamond$, be stochastically decreasing in both its arguments and be conditionally invariant (refer to Subsection \ref{ontheinnercopula} for details on conditional invariance). If, moreover, $C_2$ is NQD while $C_1$ is SD in $U_0$, then $C(u_1, u_2) > C_{Q_0(u_0)}^\diamond(u_1, u_2)$, $\forall (u_1, u_2) \in [0,1]^2$.
\end{conjecture}

The motivation of this conjecture is now given.

Let $\Delta_{u_i}(u_0)=\frac{\partial C_i(u_i, u_0)}{\partial u_0}-u_i$. Note that $\int_0^1 \frac{\partial C_i(u_i, u_0)}{\partial u_0} du_0=C_i(u_i, u_0)\biggr\rvert_0^1=u_i$, and therefore that $\int_{0}^1 \Delta_{u_i}(u_0) du_0=0$. Also note that if $C_i$ is SI in $U_0$, then $\Delta_{u_i}(u_0)$ is a decreasing function of $u_0$. Similarly, if $C_i$ is SD in $U_0$, then $\Delta_{u_i}$ is an increasing function of $u_0$. 

Using integration by parts, one can show that
\begin{equation} \label{delta1delta2}
\int_{0}^1 \Delta_{u_1}(u_0) \Delta_{u_2}(u_0) \,du_0=-\int_0^1 \bigg(C_2(u_2, u_0)-u_2u_0\bigg)\times \frac{\partial^2 C_1(u_1, u_0)}{\partial u_0^2} du_0,
\end{equation}
the integral on the left remaining positive $\forall (u_1, u_2) \in [0,1]^2$ if $C_1$ is SI in $U_0$ while $C_2$ is PQD. Likewise, the integral on the left remains negative $\forall (u_1, u_2) \in [0,1]^2$ if $C_1$ is SI in $U_0$ while $C_2$ is NQD.

Let us now write the Taylor expansion for a bivariate copula $C_{Q_0(u_0)}^\diamond$:

\begin{multline} \label{taylorexp}
C_{Q_0(u_0)}^\diamond\bigg(u_1+\Delta_{u_1}(u_0), u_2+\Delta_{u_2}(u_0)\bigg) \\
=C_{Q_0(u_0)}^\diamond(u_1, u_2)+\frac{\partial C_{Q_0(u_0)}^\diamond(u_1, u_2)}{\partial u_1}\Delta_{u_1}(u_0)+\frac{\partial C_{Q_0(u_0)}^\diamond(u_1, u_2)}{\partial u_2}\Delta_{u_2}(u_0)\\
+\frac{\partial^2 C_{Q_0(u_0)}^\diamond(u_1, u_2)}{\partial u_1 \partial u_2}\Delta_{u_1}(u_0) \Delta_{u_2}(u_0)+\frac{\partial^2 C_{Q_0(u_0)}^\diamond(u_1, u_2)}{\partial u_1^2}\Delta^2_{u_1}(u_0)/2 \\
+\frac{\partial^2 C_{Q_0(u_0)}^\diamond(u_1, u_2)}{\partial u_2^2}\Delta^2_{u_2}(u_0)/2+R\bigg(u_1+\Delta_{u_1}(u_0), u_2+\Delta_{u_2}(u_0)\bigg).
\end{multline}
If one neglects the remainder term $R\bigg(u_1+\Delta_{u_1}(u_0), u_2+\Delta_{u_2}(u_0)\bigg)$, then one can write that the outer copula in Equation \eqref{EOFC2d} is, taking into account that $C_{Q_0(u_0)}^\diamond(u_1, u_2)=C^\diamond(u_1, u_2)$ is conditionally invariant,
% \begin{equation} \nonumber
% \frac{\partial^{l_1+l_2} C_{Q_0(u_0)}^\diamond(u_1, u_2)}{\partial u_1^{l_1}\partial u_2^{l_2}}\approx0,
% \end{equation}
% $\forall l_1 > 2, \forall l_2 > 2$, and $\forall (l_1+l_2) > 2$

\begin{multline}  \nonumber
C(u_1, u_2)\approx C^\diamond(u_1, u_2) \, du_0+ \frac{\partial^2 C^\diamond(u_1, u_2)}{\partial u_1 \partial u_2}\int_0^1 \Delta_{u_1}(u_0) \Delta_{u_2}(u_0) \, du_0 \\
+ \frac{\partial^2 C^\diamond(u_1, u_2)}{\partial u_1^2}\int_0^1 \frac{\Delta^2_{u_1}(u_0)}{2} \, du_0+\frac{\partial^2 C^\diamond(u_1, u_2)}{\partial u_2^2}\int_0^1 \frac{\Delta^2_{u_2}(u_0)}{2}\, du_0.
\end{multline}

Note that both
\begin{equation}
\frac{\partial^2 C^\diamond(u_1, u_2)}{\partial u_1^2}\int_0^1 \frac{\Delta^2_{u_1}(u_0)}{2} \, du_0
\end{equation}
and
\begin{equation}
\frac{\partial^2 C^\diamond(u_1, u_2)}{\partial u_2^2}\int_0^1 \frac{\Delta^2_{u_2}(u_0)}{2}\, du_0
\end{equation}
are positive since the inner copula $C^\diamond$ is stochastically decreasing in both its arguments and $\int_0^1 \frac{\Delta^2_{u_i}(u_0)}{2} \, du_0$ is necessarily positive $\forall i$.

Moreover, by Equation \eqref{delta1delta2},
\begin{equation} \nonumber
\int_0^1 \Delta_{u_1}(u_0) \Delta_{u_2}(u_0) \, du_0=- \int_0^1 \bigg(C_2(u_2, u_0)-u_2u_0\bigg)\times \frac{\partial^2 C_1(u_1, u_0)}{\partial u_0^2} \,du_0.
\end{equation}
is a positive quantity, since $C_2$ is NQD and $C_1$ is SD in $U_0$. Therefore the outer copula can be approximated by $C^\diamond(u_1, u_2)$ plus some positive quantity, $\forall (u_1, u_2) \in [0,1]^2$, leading to the claim of the conjecture that $C(u_1, u_2) > C^\diamond(u_1, u_2)$. 

It is also possible to express a similar conjecture for $C(u_1, u_2) < C^\diamond(u_1, u_2)$:

\begin{conjecture} \label{taylordown}
Let the inner copula in \eqref{EOFC2d}, $C_{Q_0(u_0)}^\diamond=C^\diamond$, be stochastically increasing in both its arguments and be conditionally invariant. If, moreover, $C_2$ is NQD while $C_1$ is SI in $U_0$, then $C(u_1, u_2) < C^\diamond(u_1, u_2)$, $\forall (u_1, u_2) \in [0,1]^2$.
\end{conjecture}

The motivation is almost exactly the same as that of Conjecture \ref{taylorup} and is therefore left to the reader.

Of course, one might wonder at this point how reasonable it is to neglect the remainder term in \eqref{taylorexp}. In what follows, an attempt at detecting at least one example where Conjecture \ref{taylordown} fails is presented.

Let the inner copula be one of the 9 following copulas: 
\begin{itemize}
  \setlength{\itemsep}{1pt}
  \setlength{\parskip}{0pt}
  \setlength{\parsep}{0pt}
\item a Farlie-Gumbel-Morgenstern \citep{nelsen_introduction_2006, balakrishnan2009continuous} copula with a parameter's value of 0.25, 0.5 or 1,
\item a Plackett copula \citep{kemp1992continuous, nelsen_introduction_2006} with a parameter's value of 3, 12 or 64 (ranging from low to high level of dependence),
\item or a Frank copula with a parameter's value of 2.5, 6 or 14 (again, in an effort to have low, moderate and high dependence). 
\end{itemize}

All these inner copulas are conditionally invariant and SI in both arguments as required by Conjecture \ref{taylordown}. Hereafter is the proof that the proposed FGM copulas are SI in both arguments.

The expression of a FGM copula is
\begin{equation} \nonumber
C_{\theta_{\text{FGM}}}(u_1, u_2)=u_1u_2+\theta_{\text{FGM}} u_1u_2 (1-u_1)(1-u_2),
\end{equation}
where $\theta_{\text{FGM}} \in [-1, 1]$. 

One can calculate that
\begin{equation} \nonumber
\frac{\partial^2 C_{\theta_\text{FGM}}}{\partial u_1^2}=2 \theta_\text{FGM} \times(u_2 (u_2-1)),
\end{equation}
and that
\begin{equation} \nonumber
\frac{\partial^2 C_{\theta_\text{FGM}}}{\partial u_2^2}=2 \theta_\text{FGM} \times(u_1 (u_1-1)).
\end{equation}

In order for the copula to be SI in both its arguments, these two expressions have to be negative. Obviously, this will be the case as long as $\theta_\text{FGM}$ is positive.

Let now $C_1$ be one of the following copulas:
\begin{itemize}
  \setlength{\itemsep}{1pt}
  \setlength{\parskip}{0pt}
  \setlength{\parsep}{0pt}
\item a Plackett copula with a parameter's value of 3, 12 or 64,
\item a Frank copula with a parameter's value of 2.5, 6 or 14,
\item or a FGM copula with a parameter's value of 0.25, 0.5, 1.
\end{itemize}

These copulas are all SI in (at least) their second argument, as requested by Conjecture \ref{taylordown}.

Finally, let $C_2$ be one of the following 9 copulas:
\begin{itemize}
  \setlength{\itemsep}{1pt}
  \setlength{\parskip}{0pt}
  \setlength{\parsep}{0pt}
\item a Mardia copula with a parameter's value of -0.25, -0.5 or -0.75,
\item a Frank copula with a parameter's value of -2.5, -6 or -14,
\item or a FGM copula with a parameter's value of -0.25, -0.5 or -1.
\end{itemize}

These copulas are all NQD, as required by Conjecture \ref{taylordown}. A Mardia copula \citep{mardia1970families, nelsen_introduction_2006} is defined as

\begin{multline} \nonumber
C_{\theta_{\text{Mardia}}}(u_1, u_2)=\frac{\theta_{\text{Mardia}}^2 (1+\theta_{\text{Mardia}})}{2}M(u_1, u_2)\\
+(1-\theta_{\text{Mardia}}^2)\Pi(u_1, u_2)+\frac{\theta_{\text{Mardia}}^2 (1-\theta_{\text{Mardia}})}{2} W(u_1, u_2),
\end{multline}
where $M(u_1, u_2)$ is the upper Fréchet–Hoeffding bound and $W(u_1, u_2)$ the lower one, with $\theta_{\text{Mardia}} \in [-1, 1]$. A Mardia copula is NQD as long as $\theta_{\text{Mardia}}<0$ \citep[p. 188]{nelsen_introduction_2006}.

In total, there are $9^3=729$ combinations or setups of interest for the 27 presented copulas. For each of these 729 combinations, it is checked if $C(u_1, u_2)< C^\diamond(u_1, u_2)$, that is, if the outer copula is below the inner one, for a grid of 100 points in $[0,1]^2$. Since $C(u_1, u_2)$ is calculated through numerical integration, one must take into account the error on $C(u_1, u_2)$ before being able to conclude that $C(u_1, u_2)< C^\diamond(u_1, u_2)$. This is performed through the following hypothesis test:

\begin{table}[h]
\centering
\small
\begin{tabular}{ll}
$H_0:$& $C(u_1, u_2)\geq C^\diamond(u_1, u_2)$, \\
$H_1:$& $C(u_1, u_2)<C^\diamond(u_1, u_2)$,
\end{tabular}
\end{table}
\noindent with the following test statistic: 
\begin{equation} \nonumber
\frac{\hat{C}(u_1, u_2)-C^\diamond(u_1, u_2)}{\text{error}} \begin{smallmatrix}
H_0 \\ \sim \\
\end{smallmatrix}
N(0,1),
\end{equation}
and where the p-value is the surface on the left of this test statistic.

Since for a given a setup, the test has to be run over a grid of 100 points, 100 p-values is obtained for each of the $9^3$ setups. To summarize all these p-values, the average p-value is calculated for each setup.

\begin{table}[]
\centering
\begin{tabular}{c}

\includegraphics[width=1\textwidth]{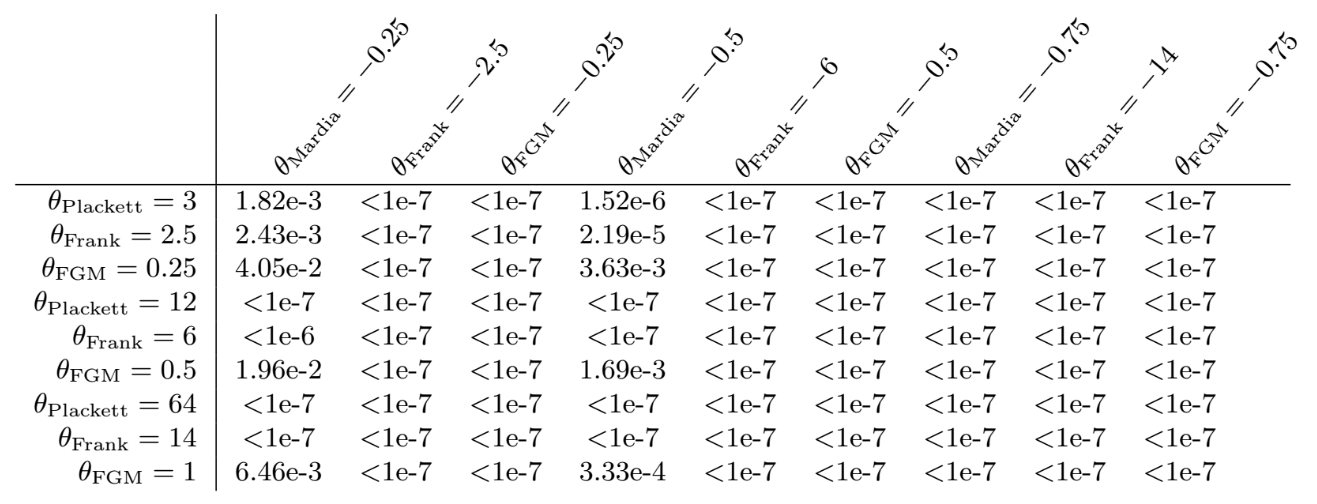}
\end{tabular}
\caption{Averaged p-values for various setups when the inner copula is a Frank copula with a parameter's value of 14. Rows are the various $C_1$ copulas, while columns describe the $C_2$ copulas. \label{Frank14inner}}
\end{table}

\textbf{Results:} For all setups except the one using as inner copula a Frank copula with a parameter's value of 14, the average p-value was less than 1e-7. Results for the inner Frank copula with a parameter's value of 14 are reported in Table \ref{Frank14inner}, where the rows are the various $C_1$ presented earlier, while the columns are the various $C_2$ presented earlier.

\subsection{Identifiability}

A model is said to be identifiable if it is theoretically possible to learn the true values of this model's underlying parameters after obtaining an infinite number of observations from it. Mathematically, this is equivalent to saying that different values of the parameters must generate different probability distributions of the observable variables. That is, if $(P_{\boldsymbol{\theta}} : \boldsymbol{\theta} \in \boldsymbol{\Theta})$ is a statistical model, with $P_{\boldsymbol{\theta}}$ a probability measure on a fixed space, the model is identifiable if $\boldsymbol{\theta}_1 \ne \boldsymbol{\theta}_2$ implies that $P_{\boldsymbol{\theta}_1} \ne P_{\boldsymbol{\theta}_2}$. 

Models built using Equations \eqref{OFC:finaleq}, \eqref{EOFC:finaleq} or \eqref{NEOFC:finaleq} are not necessarily identifiable.

\begin{example} \label{fgmidentifiability}
In Equation \eqref{OFC:finaleq}, let $d=2$ and $C_{1}, C_{2}$ be Farlie-Gumbel-Morgenstern copulas, that is,
\begin{equation*}
C_{i}(u_i,u_0;\theta_i)=u_i u_0 + \theta_i u_i u_0 (1-u_0) (1-u_i),
\end{equation*}
with $\theta_i \in [1, -1]$. 
The outer copula can be showed to be
\begin{equation*}
C(u_1, u_2;\theta_1,\theta_2)=\frac{u_1 u_2}{3}  (\theta_1 \theta_2(u_1-1)(u_2-1)+3)=u_1u_2+\frac{\theta_1 \theta_2}{3} u_1 u_2 (1-u_1)(1-u_2).
\end{equation*}
Thus, one can see that
\begin{equation*}
C(u_1, u_2;\theta_1,\theta_2)=C(u_1, u_2;\theta_1',\theta_2')
\end{equation*}
whenever $\theta_1\theta_2=\theta_1'\theta_2'$, and the last equation can be satisfied even if $(\theta_1,\theta_2)\ne(\theta_1',\theta_2')$.
\end{example}

\begin{example} \label{exnonident}
In Equation \eqref{OFC:finaleq}, let $C_i$ be the independance copula for $i \in \{2, \ldots, d\}$.
Then, the outer copula is
\begin{equation*}
C(\boldsymbol{u};\theta_1)=\int_0^1 \frac{\partial C_1(u_1, u_0 ; \theta_1)}{\partial u_0} \,u_2 \ldots u_d\, du_0,
\end{equation*}
which can be simplified to
\begin{equation*}
C(\boldsymbol{u};\theta_1)=u_2 \ldots u_d \bigg(C_1(u_1, 1;\theta_1)-C_1(u_1, 0;\theta_1)\bigg)=u_1 \ldots u_d.
\end{equation*}
The model is thus not identifiable: all values of $\theta_1$ lead to the same model, the independence copula.
\end{example}

\begin{example} \label{normalidenti}
In Equation \eqref{EOFC:finaleq}, let $d=2$, $C_1$ be a Gaussian copula with correlation $\Delta_1$ and $C_2$ be a Gaussian copula with correlation $\Delta_2$. Moreover, let $C^\diamond_{Q_0(u_0)}$ be a Gaussian copula with correlation $\rho_A$. Then, it can be shown that the outer copula is a Gaussian copula with correlation 

\begin{equation} \nonumber
\rho_A \sqrt{1-\Delta_1^2} \sqrt{1-\Delta_2^2} + \Delta_1 \Delta_2.
\end{equation}

Even in the case where $\Delta_1=\Delta_2=\Delta$, this becomes

\begin{equation} \nonumber
\rho_A \big(1-\Delta^2\big) + \Delta^2,
\end{equation}
for which one can easily find $(\rho_A, \Delta)$ and $(\rho_A', \Delta')$, $(\rho_A, \Delta) \neq (\rho_A', \Delta')$ such that $\rho_A \big(1-\Delta^2\big) + \Delta^2=\rho_A' \big(1-(\Delta')^2\big) + (\Delta')^2$. Refer to Section \ref{modelsexploration} for more details on Gaussian copula built using Equation \eqref{EOFC:finaleq}.
\end{example}

In order to shed more light on identifiability issues, a few case studies based on the Fisher information matrix are presented. The Fisher information matrix can indeed be used to check if a model is locally identifiable, see \cite{rothenberg1971identification} or \cite{iskrev2010evaluating}. If the model is not locally identifiable all over its parameter space, then it is not, in general, identifiable in that parameter space. Let $c(\boldsymbol{u};\boldsymbol{\theta})$ be the density of an EOFC. Then, the related Fisher information is defined as 

\begin{equation} \label{Fisher}
\mathcal{I}(\boldsymbol{\theta})=E\bigg(\frac{\partial \log (c(\boldsymbol{u};\boldsymbol{\theta}))}{\partial \boldsymbol{\theta}} \frac{\partial \log(c(\boldsymbol{u};\boldsymbol{\theta}))}{\partial \boldsymbol{\theta}^T}\bigg).
\end{equation}

Checking if a model is locally identifiable everywhere in $\boldsymbol{\Theta}$ is the same as checking where in $\boldsymbol{\Theta}$ this matrix is singular \citep{iskrev2010evaluating}.

Numerical computation of this matrix at a point $\boldsymbol{\theta}_r \in \boldsymbol{\Theta}$ can be done in various ways. First, one can generate $n$ observations $\boldsymbol{u}_1, \ldots, \boldsymbol{u}_n$ from the target model at point $\boldsymbol{\theta}_r$ using Algorithm \ref{data_gen} and then compute 

\begin{equation} \nonumber
\frac{1}{n}\sum_{i=1}^{n}\bigg(\frac{\partial \log (c(\boldsymbol{u}_i;\boldsymbol{\theta}))}{\partial \boldsymbol{\theta}} \frac{\partial \log(c(\boldsymbol{u}_i;\boldsymbol{\theta}))}{\partial \boldsymbol{\theta}^T}\bigg)\biggr\rvert_{\boldsymbol{\theta}=\boldsymbol{\theta}_r},
\end{equation}
where  $\frac{\partial \log (c(\boldsymbol{u}_i;\boldsymbol{\theta}))}{\partial \boldsymbol{\theta}}\bigr\rvert_{\boldsymbol{\theta}=\boldsymbol{\theta}_r}$  is numerically approached. 

So, for instance, if $\boldsymbol{\theta}_r=(\theta_{1r}, \theta_{2r}, \ldots)^T$, we have

\begin{equation} \nonumber
\frac{\partial \log(c(\boldsymbol{u}_i;\theta_1))}{\partial \theta_1}\bigr\rvert_{\theta_1=\theta_{1r}}\approx\frac{\log(c(\boldsymbol{u}_i;\theta_{1r}+h))-\log(c(\boldsymbol{u}_i;\theta_{1r}-h))}{2h},
\end{equation}
for $h$ set to an arbitrary small value.

Another approach to compute the Fisher information in \eqref{Fisher} is through numerical integration. Indeed,

\begin{equation} \nonumber
\mathcal{I}(\boldsymbol{\theta}_r)=\int_{[0,1]^d} \bigg(\frac{\partial \log (c(\boldsymbol{u};\boldsymbol{\theta}))}{\partial \boldsymbol{\theta}} \frac{\partial \log(c(\boldsymbol{u};\boldsymbol{\theta}))}{\partial \boldsymbol{\theta}^T}\bigg)\bigr\rvert_{\boldsymbol{\theta}=\boldsymbol{\theta}_r} c(\boldsymbol{u};\boldsymbol{\theta}_r) d\boldsymbol{u},
\end{equation}
where  $\frac{\partial \log (c(\boldsymbol{u};\boldsymbol{\theta}))}{\partial \boldsymbol{\theta}}\bigr\rvert_{\boldsymbol{\theta}=\boldsymbol{\theta}_r}$ has to be numerically approached, too.

Note that in both approaches, the computation of $c(\boldsymbol{u};\boldsymbol{\theta}_r)$ itself requires numerical integration, which is performed, in this paper, using adaptative multidimensional integration \citep{cubaR} or naive Monte Carlo integration.

\textbf{Case study 1 to 3.} Let $C_1=C_2$ be a Frank copula with a parameter's value such that $\tau_{\text{Frank}}=0.25, 0.5$ and 0.75. The inner copula in all three cases is a normal copula with correlation $\theta$. The model turned out to be locally identifiable for each of the 3 cases. Table \ref{case_final1} displays the results for case 3.

\begin{table}[H]
\centering
 \footnotesize{
 \begin{tabular}{r|cccccccccc}
 $\theta$ & 0.09 & 0.18 & 0.27 & 0.36 & 0.45 & 0.55 & 0.64 & 0.73 & 0.82  & 0.91  \\
 \hline
 $\hat{\mathcal{I}}(\theta)$      & 0.595 & 0.745 & 0.949 & 1.246 & 1.698 & 2.438 & 3.908 &  6.762 & 15.617 & 58.255 \\

 \end{tabular}}
 \caption{The Fisher information as a function of $\theta$ for case study 3. \label{case_final1}}
 \end{table}

\textbf{Case study 4.} In this case study, the inner copula is the independence one, $C_2$ is a Gumbel copula with a fixed parameter value corresponding to a moderate dependence and $C_1$ is the copula with the parameter of interest, $\theta_{\text{Gumbel}}$. Table \ref{case9} shows a rapid decrease of the Fisher information as the parameter of interest increases.

\begin{table}[H]
\centering
\footnotesize{
\begin{tabular}{r|cccccccccc}
$\tau(\theta_{\text{Gumbel}})$ & 0.09 & 0.18 & 0.27 & 0.36 & 0.45 & 0.55 & 0.64 & 0.73 & 0.82  & 0.91  \\
\hline
 $\hat{\mathcal{I}}(\theta_{\text{Gumbel}})$          & 1.30 & 0.73 & 0.43 & 0.24 & 0.13 & 0.06 & 0.02 & <1e-2 & <1e-3 & <1e-5
\end{tabular}}
\caption{The Fisher information for case study 4. \label{case9}}
\end{table}

\textbf{Case study 5.} In this last case study, $C_1$ is the independence copula, $C_2$ is the Frank copula with parameter $\theta_{\text{Frank}}$ and the inner copula is the normal copula with parameter $\theta_{\text{normal}}$. The Fisher information is this times a $2 \times 2$ matrix. Figure \ref{last_case} shows the determinant of the matrix as a function of $\theta_{\text{Frank}}$, expressed as a Kendall's tau for convenience, and $\theta_{\text{normal}}$.

\begin{figure}[H]
\centering
\begin{tabular}{c}

\includegraphics[width=0.55\textwidth]{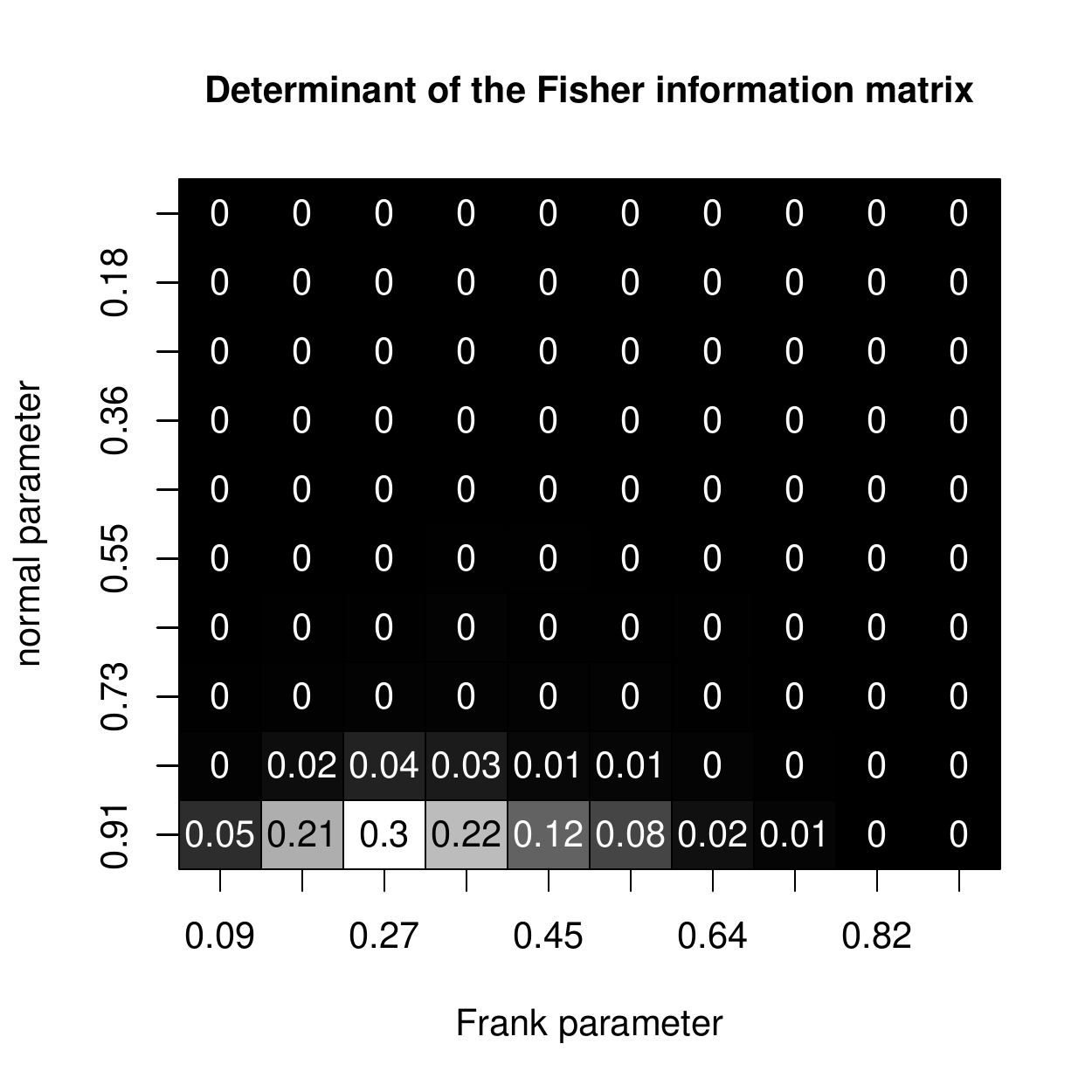}
\end{tabular}
\caption{The determinant of the Fisher information matrix for case study 5. \label{last_case}}
\end{figure}

As one can see, the model is locally identifiable in a narrow part of the space of the parameters, picking at around (0.91, 0.27).

\section{Data generation}

To generate one realization $(u_1,\dots,u_d)$ of the random vector $(U_1,\dots,U_d)$ with distribution $C$ given by \eqref{eq:representation 2}, one takes the second part of $(u_0,u_1,\dots,u_d)$, from $u_1$ to $u_d$, a realization of $(U_0,U_1,\dots,U_d)$, where $U_0$ is the latent factor. Remembering that, given $U_0=u_0$, the distribution of $(U_1,\dots,U_d)$ can be split into the inner copula $C_{u_0}^\diamond$ and a set of univariate margins $\{C_{i\vert 0}(u_i\vert u_0)\}=\{\frac{\partial C_i(u_i, u_0)}{\partial u_0}\}$, with $C_{i \vert 0}^{-1}(\bullet\vert u_0)$ denoting the inverse function, $i \in \{1,\dots,d\}$, the following algorithm produces the desired output.
\begin{algorithm}[H]
\caption{Generating one observation from \eqref{eq:representation 2}. \label{data_gen}}
\begin{algorithmic}[1]
\State Generate one observation $u_0$ from a standard uniform random variable.
\State Generate one observation $(u^\diamond_{1}, \ldots, u^\diamond_{d})$ from $C_{u_0}^\diamond$.
\State Put $u_i=C_{i \vert 0}^{-1}(u^\diamond_{i}\vert u_0)$ for $i=1,\dots,d$.
\end{algorithmic}
\end{algorithm}
\noindent Note that, in the presence of conditional invariance, step 1 in the above algorithm is not required for step 2. Needless to say, in the first step, one could have sampled from $F_0$, the distribution of $X_0$, and in the second step, one would have sampled from $C_{x_0}$ instead of $C_{u_0}^\diamond$. 

Data generation for a NEOFC can be performed using a generalized version of Algorithm \ref{data_gen}. The main hurdle for this generalization is to define $G_i^{-1}$, such that 
\begin{equation*}
G_i^{-1}(G_i(u_i ; t_w, \ldots t_1) ; t_w, \ldots, t_1)=u_i, 
\end{equation*}
or, in short, such that $G_i^{-1}(G_i(u_i))=u_i$.

Recall that
\begin{equation} \nonumber
G_i(u_i ; t_w, \ldots, t_1)=H_{iw}^{t_w} \circ \dots \circ H_{i1}^{t_1}(u_i)
\end{equation}
and that $H_{ij}^{t_j}(u_i)=\frac{\partial C_{ij}(u_i, t_j)}{\partial t_j}$.

Now define $H_{ij}^{-1,t_j}$ as the inverse function of $H_{ij}^{t_j}$, that is, 
\begin{equation} \nonumber
H_{ij}^{-1,t_j}(H_{ij}^{t_j}(u_i))=u_i.
\end{equation}

This is the same as writing
\begin{equation} \nonumber
H_{ij}^{-1,t_j}\circ H_{ij}^{t_j}(u_i)=u_i.
\end{equation}

Since $G_i$ is just a composition of the form $H_{iw}^{t_w} \circ \dots \circ H_{i1}^{t_1}(u_i)$, we can extract back $u_i$ using the composition $H_{i1}^{-1,t_1} \circ \dots \circ H_{iw}^{-1,t_w}(\bullet)$. The inverse of $G_i$ is therefore:
\begin{equation}
G_i^{-1}(\bullet)=H_{i1}^{-1,t_1} \circ \dots \circ H_{iw}^{-1,t_w}(\bullet).
\end{equation}

% Indeed,
% \begin{equation}
% G_i^{-1}(G_i(u_i))=H_{i1}^{-1,t_1} \circ \dots \circ H_{iw}^{-1,t_w}\big(H_{iw}^{t_w} \circ \dots \circ H_{i1}^{t_1}(u_i) \big)=u_i.
% \end{equation}

Now that $G_i^{-1}$ is defined, Algorithm \ref{data_gen_NEOFC}, the generalized version of Algorithm \ref{data_gen}, can be given.

\noindent {\bf Input.} A NEOFC, that is, a $d$-variate copula $C_{t_w}^{\diamond_w}$, the related mappings, $d\times w$ bivariate copulas $\{C_{ij}\}$, and the related parameters.

\noindent {\bf Output.} One observation from the input NEOFC, as a $d$-dimensional vector.

\begin{algorithm}[H]
\caption{Generating one observation from a NEOFC. \label{data_gen_NEOFC}}
\begin{algorithmic}[1]
\State Generate $w$ observations from a uniform random variable on [0,1]. Denote these values as $t_1, \ldots, t_w$.
\State Generate one observation from $C_{t_w}^{\diamond_w}$. Denote the resulting vector as $u_1^{\diamond_w} \ldots u_d^{\diamond_w}$.

\State On each $u_i^{\diamond_w}$ in $\{u_1^{\diamond_w}, \ldots, u_d^{\diamond_w}\}$, run the corresponding function $G_i^{-1}(\bullet ; t_w, \ldots, t_1)$. Denote the result as $u_i$.
\end{algorithmic}
\end{algorithm}
The end result of the last step of Algorithm \ref{data_gen_NEOFC} is a vector $(u_1, \ldots, u_d)$, one observation from the input NEOFC. Run Algorithm \ref{data_gen_NEOFC} multiple times to get multiple observations.

Critical to both algorithms in this section is the inversion of $\frac{\partial C_{i}(u_i, u_0)}{\partial u_0}$ with respect to $u_i$ in the case of an EOFC or of $\frac{\partial C_{ij}(u_i, t_j)}{\partial t_j}$ with respect to $u_i$ in the case of a NEOFC. This is usually not hard to get symbolically. In the worst cases, $C_{i \vert 0}^{-1}$ or $H_{ij}^{-1,t_j}$ can be get numerically.

% In the appendix chapter of the thesis, R code allowing to work with OFCs, EOFCs and NEOFCs is showcased, including a function allowing to generate observations according to the two algorithms presented in this section. At the moment, the inverse of $\frac{\partial C_{i}(u_i, u_0)}{\partial u_0}$ with respect to $u_i$ is implemented for more than 10 bivariate copulas. See the appendix chapter for a detailed list.

\section{Models} \label{modelsexploration}
In this section, several models built from Equation \eqref{eq:representation 2} or Equation \eqref{NEOFC} are presented, including many well-known copula models from the literature.

\begin{paragraph}{Krupskii-Huser-Genton copulas.}
In their paper, \cite{krupskiicopulas} offer a new family of one-factor copulas that can be used to model replicated spatial data. Let $C_A$ be a $d$-variate Gaussian copula with correlation matrix $A=\{\rho_{ij}\}$. Let 
\begin{equation} \nonumber
F_\bullet(x_i)=\int_{-\infty}^{+\infty}\Phi(x_i-x_0)dF_0(x_0),
\end{equation}
where $F_0$ is some univariate CDF, while $\Phi$ is the CDF of the univariate standard Gaussian distribution. Then, the KHG family of copulas corresponds to an EOFC where $C_{Q_0(u_0)}^\diamond=C^\diamond=C_A$ and $\frac{\partial C_i(u_i, u_0)}{\partial u_0}=\Phi(F_\bullet^{-1}(u_i)-F_0^{-1}(u_0))$, $\forall i \in \{1, \ldots d\}$.

As shown by \cite{krupskiicopulas}, tail dependence and tail asymmetry  for this model is possible, based on the choice of $F_0$ and $A$.

Let $C^{[i,j]}(u_i, u_j)$, $i \neq j$, be a bivariate margin of $C$, the outer copula. Recall that the lower tail dependence of a bivariate copula, such as $C^{[i,j]}(u_i, u_j)$, is \citep{hartmann2004asset, mcneil2015quantitative}:

\begin{equation} \nonumber
\lambda^L_{ij}=\lim_{q\rightarrow0}\frac{C^{[i,j]}(q, q)}{q},
\end{equation}
while the upper one is defined as
\begin{equation} \nonumber
\lambda^U_{ij}=\lim_{q\rightarrow0}\frac{2q-1+C^{[i,j]}(1-q, 1-q)}{q}.
\end{equation}

As an example, if
\begin{equation} \nonumber
F_0(x_0)=1-K x_0^\beta e^{-\theta x_0^\alpha},
\end{equation}
\cite{krupskiicopulas} showed that for any bivariate margin $C^{[i,j]}(u_i, u_j)$ of $C$, as long as $\theta, K >0$, $0<\alpha<1$ and $\beta \in \mathbb{R}$, or as long as $\theta, K >0$, $\alpha=0$ and $\beta <0$, 

\begin{equation} \nonumber
\lambda^U_{ij}=1.
\end{equation}

A more interesting case however arises if $\theta, K >0$, $\alpha=1$ and $\beta \in \mathbb{R}$. Then, 
\begin{equation} \nonumber
\lambda^U_{ij}=2 \Phi \bigg[-\theta \sqrt{\frac{1-\rho_{ij}}{2}} \bigg], 
\end{equation}
meaning that tail asymmetry becomes possible through the correlation matrix $A=\{\rho_{ij}\}$.

The upper tail dependence coefficient can also be made such that $\lambda^U_{ij}=0$: one need to let $\theta, K >0$, $\alpha>1$ and $\beta \in \mathbb{R}$.

For more examples and details, refer to \cite{krupskiicopulas}.
\end{paragraph}

\begin{paragraph}{Farlie-Gumbel-Morgenstern copulas.} Let the inner copula of a bidimensional EOFC be such that $C^\diamond_{Q_0(u_0)}=\Pi$, while $C_1, C_2$ are Farlie-Gumbel-Morgenstern copulas with parameters $\theta_1$ and $\theta_2$. Then, the resulting outer copula is a FGM copula with parameter $\frac{\theta_1 \theta_2}{3}$. Although a bit tedious, the proof is trivial. Also refer to Example \ref{fgmidentifiability}. 

Note that in a bivariate FGM with parameter $\theta \in [-1, 1]$, the related Kendall's tau is $\tau=\frac{2\theta}{9}$ \citep[p. 162]{nelsen_introduction_2006}. This means that even if the parameters $\theta_1$ and $\theta_2$ of $C_1$ and $C_2$ are set to their upper limit, the parameter of the resulting FGM will be 1/3, corresponding to a rather low Kendall's tau of 0.074. The model where $C^\diamond_{Q_0(u_0)}=\Pi$, while $C_1, C_2$ are FGM copulas, is thus trapped between a correlation of -0.074 and 0.074 only. This can be fixed using EOFCs.

In an EOFC, the inner copula can be viewed as the baseline level of dependence. Instead of letting $C^\diamond_{Q_0(u_0)}=\Pi$, one could let $C^\diamond_{Q_0(u_0)}=M$, where $M$ is the upper Fréchet-Hoeffding bound. In this scheme, the two FGM copulas $C_1, C_2$ can only be used to decrease the dependence between the two random variables of interest, $U_1$ and $U_2$. Figure \ref{fgmfig} shows the empirical Kendall's tau calculated (sample size = 1000) for various combinations of $\theta_1$ and $\theta_2$ when the inner copula is $M$. As one can see, the improved model is now trapped between a correlation of $\sim0.55$ and 1. This range can thus be tuned by changing the inner copula.

\begin{figure}
\centering
\begin{tabular}{c}

\includegraphics[width=0.77\textwidth]{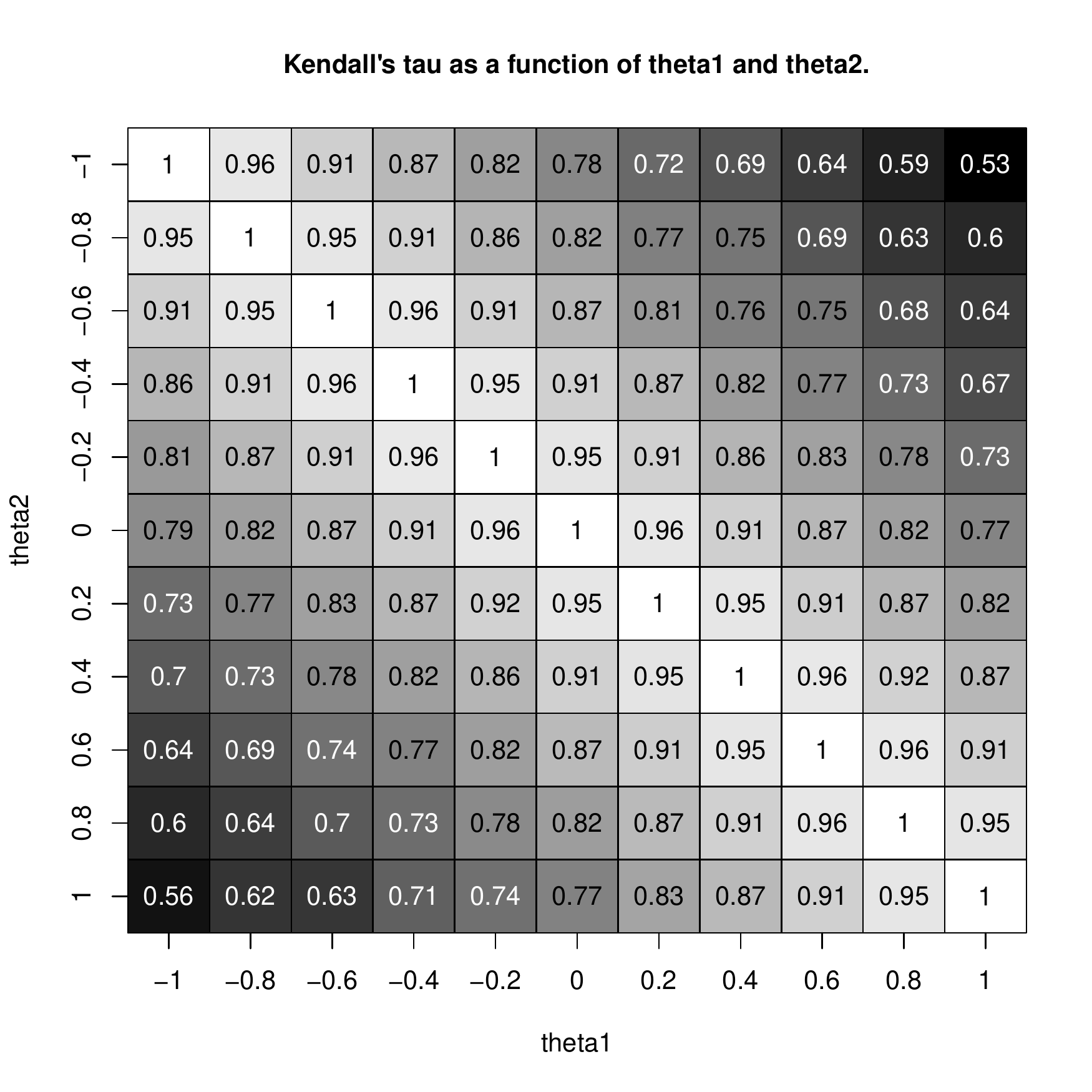}
\end{tabular}
\caption{Empirical Kendall's tau for an EOFC where the inner copula is $M$ and $C_1, C_2$ are FGM copulas, described through $\theta_1 \in [-1, 1]$ and $\theta_2 \in [-1, 1]$. \label{fgmfig}}
\end{figure}
\end{paragraph}

\begin{paragraph}{Archimedean copulas.}
Let $\psi$ be a  completely monotonic function on $[0,\infty]$, that is,
$(-1)^kd^k/dt^k\psi(t)\ge 0$ for all integers $k$ and all $t>0$, and such that $\psi(0)=1$ while
\begin{equation*}
\psi(\infty)=\lim_{t\to\infty}\psi(t)=0. 
\end{equation*}

If a copula $C$ can be written as 
\begin{equation*}
C(u_1,\dots,u_d)=\psi(\psi^{-1}(u_1)+\dots+\psi^{-1}(u_d)),
\end{equation*}
then it is called an Archimedean copula with generator $\psi$ \citep{doi:10.1080/00949650701255834}. Let us note that, in order to make sure $C$ is a proper copula, the above-mentioned conditions on $\psi$ are sufficient, but not necessary. For sufficient and necessary conditions, see \cite{2009arXiv0908.3750M}.

\begin{proposition}
\label{prop:archimedean}
In \eqref{eq:representation 2}, let $C_{x_0}=\Pi$, assume that the support of $X_0$ is $[0,\infty]$, and put
\begin{equation*}
C_{i}(u_i,u_0)=\int_0^{Q_0(u_0)} e^{-t\psi^{-1}(u_i)} f_0(t) \, dt, 
\end{equation*}
where
\begin{equation*}
\psi(x)=\int_0^\infty e^{-t x} f_0(t) \, dt,
\end{equation*}
$i\in \{1,\dots,d\}$, with $f_0$ being the derivative of $F_0$ and $Q_0$ its inverse, with $Q_0(u_0)=x_0$. It can be checked that $\psi$ is completely monotonic, see for instance \cite{joe_multivariate_2001}. Then $C$, the left-hand side of equation \eqref{eq:representation 2}, is an Archimedean copula with generator $\psi$.
\end{proposition}

\begin{proof} 
Checking that $C_{i}$ is a copula is straightforward. Moreover, since
\begin{align*}
  \frac{\partial C_{i}(u_i,u_0)}{\partial u_0}
  =e^{-Q_0(u_0)\psi^{-1}(u_i)},
\end{align*}
it holds that
\begin{align*}
  C(u_1,\dots,u_d)=\int_0^\infty e^{-x_0 \, \sum_{i=1}^d \psi^{-1}(u_i)}
  f_0(x_0) \, dx_0
  =\psi \Big(\sum_{i=1}^d \psi^{-1}(u_i) \Big).
\end{align*}
\end{proof}

Note that a reformulation of the above construction can be found in \cite{joe_multivariate_2001}.
\end{paragraph}

\begin{paragraph}{Hierarchical Archimedean copulas.}
Archimedean copulas can be nested in order to get more flexible models. Hierarchical or Nested Archimedean copulas (HACs or NACs) were introduced by \cite{joe_multivariate_2001} and have been the main topic of many research papers since, see for instance \cite{doi:10.1080/00949650701255834}, \cite{hofert2013densities}, or \cite{okhrinOstap2013properties}. The simplest hiearchical Archimedean copula consists of a bivariate Archimedean copula 
\begin{equation*}
C_{12}(u_1, u_2)=\psi_{12}(\psi^{-1}_{12}(u_1)+\psi^{-1}_{12}(u_2)),
\end{equation*}
which is nested into another bivariate Archimedean copula 
\begin{equation*}
C_{123}(\bullet, u_3)=\psi_{123}(\psi^{-1}_{123}(\bullet)+\psi^{-1}_{123}(u_3))
\end{equation*} 
in order to get a copula of the form
\begin{multline} \label{nac}
C(u_1, u_2, u_3)=C_{123}(C_{12}(u_1, u_2), u_3) \\ =\psi_{123}\big(\psi_{123}^{-1}(\psi_{12}(\psi_{12}^{-1}(u_1) +\psi_{12}^{-1}(u_2)))+ \psi_{123}^{-1}(u_3)\big).
\end{multline}
In general, an arbitrary pair of generators $(\psi_{123}, \psi_{12})$ does not ensure that the copula in Equation \eqref{nac} is a proper copula. See \cite*{joe_multivariate_2001} and \cite*{doi:10.1080/00949650701255834} for more on this matter.

\begin{proposition}
\label{prop:narchimedean}
Define $\psi_{123}$ the same way $\psi$ was defined in Proposition \ref{prop:archimedean}. Also let $C_{i}$ as in Proposition \ref{prop:archimedean}. Further define
\begin{align*}
C_{x_0}(u, v, w)=\exp \Bigg(-x_0\times \nu \bigg(\nu^{-1}\Big[\frac{1}{x_0}\log \Big(\frac{1}{u}\Big)\Big]+\nu^{-1}\Big[\frac{1}{x_0}\log \Big(\frac{1}{v}\Big)\Big]\bigg)\Bigg) \times w;\\
\end{align*}
where $\nu(\bullet)=\psi_{123}^{-1}\big(\psi_{12}(\bullet)\big)$, $\nu(\bullet)^{-1}=\psi_{12}^{-1}\big(\psi_{123}(\bullet)\big)$ and $\psi_{12}(\bullet)$  is equal to the integral from $0$ to $\infty$ of $\exp(-t\bullet)dF_{12}(t)$ with $F_{12}$ some distribution function.
Then \eqref{eq:representation 2} is the copula given in \eqref{nac}.

\end{proposition}

\begin{proof}
First, it is proven that $C_{x_0}$ is a copula.
Define, for $0\leq u,v,w\leq 1$,
\begin{multline*}
G_{x_0}(u, v, w)=\exp \Bigg(-x_0\times \nu \bigg(\nu^{-1}\Big[\psi_{123}^{-1}\Big(u\Big)\Big]+\nu^{-1}\Big[\psi_{123}^{-1}\Big(v\Big)\Big]\bigg)\Bigg) \\
\times \exp\Bigg(-x_0\times \psi_{123}^{-1}\bigg(w\bigg)\Bigg).
\end{multline*}
One can check that $C_{x_0}$ in Proposition \ref{prop:narchimedean} is the copula corresponding to the multivariate distribution $G_{x_0}$. And from \cite{joe_multivariate_2001}, page 88, it can be easily deduced that $G_{x_0}$ is indeed a distribution function. Therefore $C_{x_0}$ is a proper copula.

The proof then proceeds by showing that $C$ in \eqref{eq:representation 2} is a hierarchical Archimedean copula. $C$ in \eqref{eq:representation 2} becomes
\begin{align} \label{nac_simp}
  C(u_1,u_2,u_3)= 
\int_{0}^{1}&\exp \Bigg(-Q_0(u_0)\times \nu \bigg(\nu^{-1}\Big[\frac{1}{Q_0(u_0)}\log \Big(\frac{1}{C_{1\vert 0}(u_1\vert u_0)}\Big)\Big]\\
&+\nu^{-1}\Big[\frac{1}{Q_0(u_0)}\log \Big(\frac{1}{C_{2\vert 0}(u_2\vert u_0)}\Big)\Big]\bigg)\Bigg)
\times C_{3\vert 0}(u_3\vert u_0)\, du_0. \notag
\end{align}
In the proof of Proposition \ref{prop:archimedean}, it was shown that
$\frac{\partial C_i(u_i, u_0)}{\partial u_0}=C_{i\vert 0}(u_i\vert u_0)=\exp(-Q_{0}(u_0)\times \psi_{123}^{-1}(u_i))$.
Replacing in \eqref{nac_simp} gives
\begin{align*}
&\int_{0}^{1}\exp \Bigg(-Q_0(u_0)\times \nu \bigg(\nu^{-1}\Big[\psi_{123}^{-1}(u_1)\Big]+\nu^{-1}\Big[\psi_{123}^{-1}(u_2)\Big]\bigg)\Bigg)\\
&\bigqquad\bigqquad \exp(-F_{0}^{-1}(u_0)\times \psi_{123}^{-1}(u_3)) \, du_0\\
=&\int_{0}^{\infty}  \exp \Bigg( -x_0 \bigg[ \psi_{123}^{-1}\bigg(\psi_{12}\Big(\psi_{12}^{-1}(u_1)+\psi_{12}^{-1}(u_2)\Big)\bigg) + \psi_{123}^{-1}(u_3) \bigg]\Bigg) dF_0(x_0)\\
=&\,\,\psi_{123}\big(\psi_{123}^{-1}(\psi_{12}(\psi_{12}^{-1}(u_1) +\psi_{12}^{-1}(u_2)))+ \psi_{123}^{-1}(u_3)\big),
\end{align*}
which is the NAC from \eqref{nac}.
\end{proof}

\end{paragraph}

\begin{paragraph}{Hierarchical NEOFCs.}
OFCs are not suited to model hierarchical dependence. Assume that one wishes to build a one-factor copula on $(U_1, U_2, U_3, U_4)$ such that the resulting model is described by the matrix of Kendall's taus in \eqref{impcor}.

\begin{equation} \label{impcor}
\begin{matrix}
 & U_1 & U_2 & U_3 & U_4\\ 
U_1 & & 1 & 0.25 & 0.25\\ 
U_2 &  & & 0.25 & 0.25\\ 
U_3 &  &  & & 1\\ 
U_4 &  &  &  &
\end{matrix}
\end{equation}

Since $U_1$ and $U_2$ are related through the maximum value of 1, it means that their bivariate margin must be

\begin{equation} \nonumber
C^{[1,2]}(u_1, u_2)=\min(u_1, u_2)=\int_0^1 \frac{\partial \min(u_1, u_0)}{\partial u_0} \frac{\partial \min(u_2, u_0)}{\partial u_0} du_0.
\end{equation}

Similarly, the bivariate margin of $U_3$ and $U_4$ is

\begin{equation} \nonumber
C^{[3,4]}(u_3, u_4)=\min(u_3, u_4)=\int_0^1 \frac{\partial \min(u_3, u_0)}{\partial u_0} \frac{\partial \min(u_4, u_0)}{\partial u_0} du_0.
\end{equation}

Therefore, the OFC on $(U_1, U_2, U_3, U_4)$ has to be

\begin{equation} \nonumber
C(u_1, u_2, u_3, u_4)=\int_0^1 \frac{\partial \min(u_1, u_0)}{\partial u_0}\frac{\partial \min(u_2, u_0)}{\partial u_0} \frac{\partial \min(u_3, u_0)}{\partial u_0} \frac{\partial \min(u_4, u_0)}{\partial u_0} du_0.
\end{equation}

By proposition \ref{upperF} however, this corresponds to $C(u_1, u_2, u_3, u_4)=\min(u_1,\allowbreak u_2, u_3, u_4)$, for which the correlation matrix is not \eqref{impcor}. A OFC will never be able to reproduce a matrix of Kendall's taus such as the one in \eqref{impcor}. In general, OFCs do not seem to be suited to model matrices of the form

\begin{equation} \label{impcor2}
\begin{matrix}
 & U_1 & U_2 & U_3 & U_4\\ 
U_1 & & \tau_{12} & \tau_0 & \tau_0\\ 
U_2 &  & & \tau_0 & \tau_0\\ 
U_3 &  &  & & \tau_{34}\\ 
U_4 &  &  &  &
\end{matrix}
\end{equation}
in which $\tau_{12}, \tau_{34}>\tau_0$.

On the other hand, it is possible to build NEOFCs for which the matrix of Kendall's taus is as in \eqref{impcor2}.

Start from Equation \eqref{NEOFC}, and let $w=3$, that is, we have a triple layer of linking copulas. Let $C_{t_3}^{\diamond_3}=\Pi$. For $j=3$, let all linking copulas $C_{ij}$ be a Frank copula with a same parameter $\theta_3$. This layer gives the baseline level of dependence between the four random variables of interest. Next, use the second layer, $j=2$, to couple $U_1$ and $U_2$, that is, let $C_{12}$ and $C_{22}$ be a Frank copula with a common paramater $\theta_2$, while $C_{32}$ and $C_{42}$ are set to the independence copula. Finally use the first layer, $j=1$, to couple $U_3$ and $U_4$, that is, let $C_{31}$ and $C_{41}$ be a Frank copula with a common paramater $\theta_1$, while $C_{11}$ and $C_{21}$ are set to the independence copula. Table \ref{NEOFChierarchy1} allows one to visualize this scheme.

\begin{table}[H]
\centering

\begin{tabular}{rccccl}
\multicolumn{1}{l}{}       & $i=1$                      & $i=2$                      & $i=3$                      & $i=4$                      &            \\ \cline{2-5}
\multicolumn{1}{r|}{$j=1$} & \multicolumn{1}{c|}{$\Pi$} & \multicolumn{1}{c|}{$\Pi$} & \multicolumn{1}{c|}{Frank} & \multicolumn{1}{c|}{Frank} & $\theta_1$ \\ \cline{2-5}
\multicolumn{1}{r|}{$j=2$} & \multicolumn{1}{c|}{Frank} & \multicolumn{1}{c|}{Frank} & \multicolumn{1}{c|}{$\Pi$} & \multicolumn{1}{c|}{$\Pi$} & $\theta_2$ \\ \cline{2-5}
\multicolumn{1}{r|}{$j=3$} & \multicolumn{1}{c|}{Frank} & \multicolumn{1}{c|}{Frank} & \multicolumn{1}{c|}{Frank} & \multicolumn{1}{c|}{Frank} & $\theta_3$ \\ \cline{2-5}
\multicolumn{1}{r|}{$C^\diamond$} & \multicolumn{4}{c|}{$\Pi$}                                                                                        &            \\ \cline{2-5}
\end{tabular}
\caption{A hierarchical NEOFC with three layers and an inner copula set to independence. \label{NEOFChierarchy1}}

\end{table}

\begin{example}
In Table \ref{NEOFChierarchy1}, set $\theta_3$ to a value of 5.74, corresponding to a Kendall'tau of 0.5 for the related linking copulas and $\theta_2=\theta_1$ to a value of 6.73, corresponding to a Kendall's tau of 0.55 for the related linking copulas. The empirical Kendall's taus based on 10000 observations generated from the model through Algorithm \ref{data_gen_NEOFC} are as shown hereafter.
\begin{code}
> round(cor(generated.data, method="kendall"), 3)
      [,1]  [,2]  [,3]  [,4]
[1,] 1.000 0.559 0.119 0.123
[2,] 0.559 1.000 0.118 0.121
[3,] 0.119 0.118 1.000 0.563
[4,] 0.123 0.121 0.563 1.000
\end{code}

While the correlation between $U_1$, $U_2$ and $U_3$, $U_4$ is strong, one can note that even though $\theta_3$ corresponds to a correlation of 0.5, the correlation for the couples $(U_1, U_3)$, $(U_1, U_4)$, $(U_2, U_3)$ and $(U_2, U_4)$ appears to have dampen to $\sim0.12$. Pushing $\theta_3$ towards higher values can help: if $\theta_3$ is set to a value of 14.14, corresponding to a Kendall's tau of 0.75, the empirical Kendall's taus become as shown hereafter.
\begin{code}
> round(cor(generated.data, method="kendall"), 3)
      [,1]  [,2]  [,3]  [,4]
[1,] 1.000 0.743 0.220 0.222
[2,] 0.743 1.000 0.219 0.222
[3,] 0.220 0.219 1.000 0.738
[4,] 0.222 0.222 0.738 1.000
\end{code}

This increase remains however slow and even with a high value for $\theta_3$, such as 38.28, corresponding to a Kendall's tau of 0.9, the correlation between $U_1$ and $U_3$ will not exceed 0.27.
\end{example}

% \begin{example}
% Set $\theta_3$ to an arbitrary value of 5.74, corresponding to a Kendall's $\tau$ of 0.5, $\theta_2$ to 14.14, corresponding to a Kendall's tau of 0.75, and $\theta_1$ to 2.37, corresponding to a Kendall's $\tau$ of 0.25. The empirical Kendall's taus based on 10000 observations generated from this scheme are shown hereafter.

% \begin{code}
% > round(cor(generated.data, method="kendall"), 3)
      % [,1]  [,2]  [,3]  [,4]
% [1,] 1.000 0.742 0.100 0.108
% [2,] 0.742 1.000 0.102 0.111
% [3,] 0.098 0.102 1.000 0.370
% [4,] 0.108 0.111 0.370 1.000
% \end{code}
% \end{example}

To overcome the issue of low dependence between the pairs $(U_1, U_3)$, $(U_1, U_4)$, $(U_2, U_3)$ and $(U_2, U_4)$, one can make use of the inner copula to set the baseline level of dependence, rather than using a layer of identical linking copulas. See Table \ref{NEOFChierarchy2}.

\begin{table}[H]
\centering

\begin{tabular}{rccccl}
\multicolumn{1}{l}{}       & $i=1$                      & $i=2$                      & $i=3$                      & $i=4$                      &                   \\ \cline{2-5}
\multicolumn{1}{r|}{$j=1$} & \multicolumn{1}{c|}{$\Pi$} & \multicolumn{1}{c|}{$\Pi$} & \multicolumn{1}{c|}{Frank} & \multicolumn{1}{c|}{Frank} & $\theta_1$        \\ \cline{2-5}
\multicolumn{1}{r|}{$j=2$} & \multicolumn{1}{c|}{Frank} & \multicolumn{1}{c|}{Frank} & \multicolumn{1}{c|}{$\Pi$} & \multicolumn{1}{c|}{$\Pi$} & $\theta_2$        \\ \cline{2-5}
\multicolumn{1}{r|}{$C^\diamond$} & \multicolumn{4}{c|}{Frank}                                                                                        & $\theta^\diamond$ \\ \cline{2-5}
\end{tabular}
\caption{A hierarchical NEOFC with 2 layers and an inner copula set to a Frank copula. \label{NEOFChierarchy2}}

\end{table}

\begin{example}
In Table \ref{NEOFChierarchy2}, set $\theta^\diamond$ to a value of 14.14, corresponding to a Kendall's $\tau$ of 0.75, $\theta_2$ to 1.38, corresponding to a Kendall's tau of 0.15, and $\theta_1$ to 6.73, corresponding to a Kendall's $\tau$ of 0.55. The empirical Kendall's taus based on 10000 observations generated from the related hierarchical NEOFC are shown hereafter.

\begin{code}
> round(cor(generated.data, method="kendall"), 3)
      [,1]  [,2]  [,3]  [,4]
[1,] 1.000 0.755 0.388 0.385
[2,] 0.755 1.000 0.387 0.385
[3,] 0.388 0.387 1.000 0.809
[4,] 0.385 0.385 0.809 1.000
\end{code}
\end{example}

As last example of a hierarchical NEOFC, we revisit the tree on the right of Figure~11 in \cite{segers2014nonparametric}, this tree corresponding to an attempt to map the dependencies between daily log returns using hierarchical Archimedean copulas. Figure~11 in \cite{segers2014nonparametric} is, for convenience, reproduced in this paper as Figure \ref{wanted}.

\begin{figure}[H]
\centering
\begin{tabular}{cc}

\includegraphics[width=0.49\textwidth]{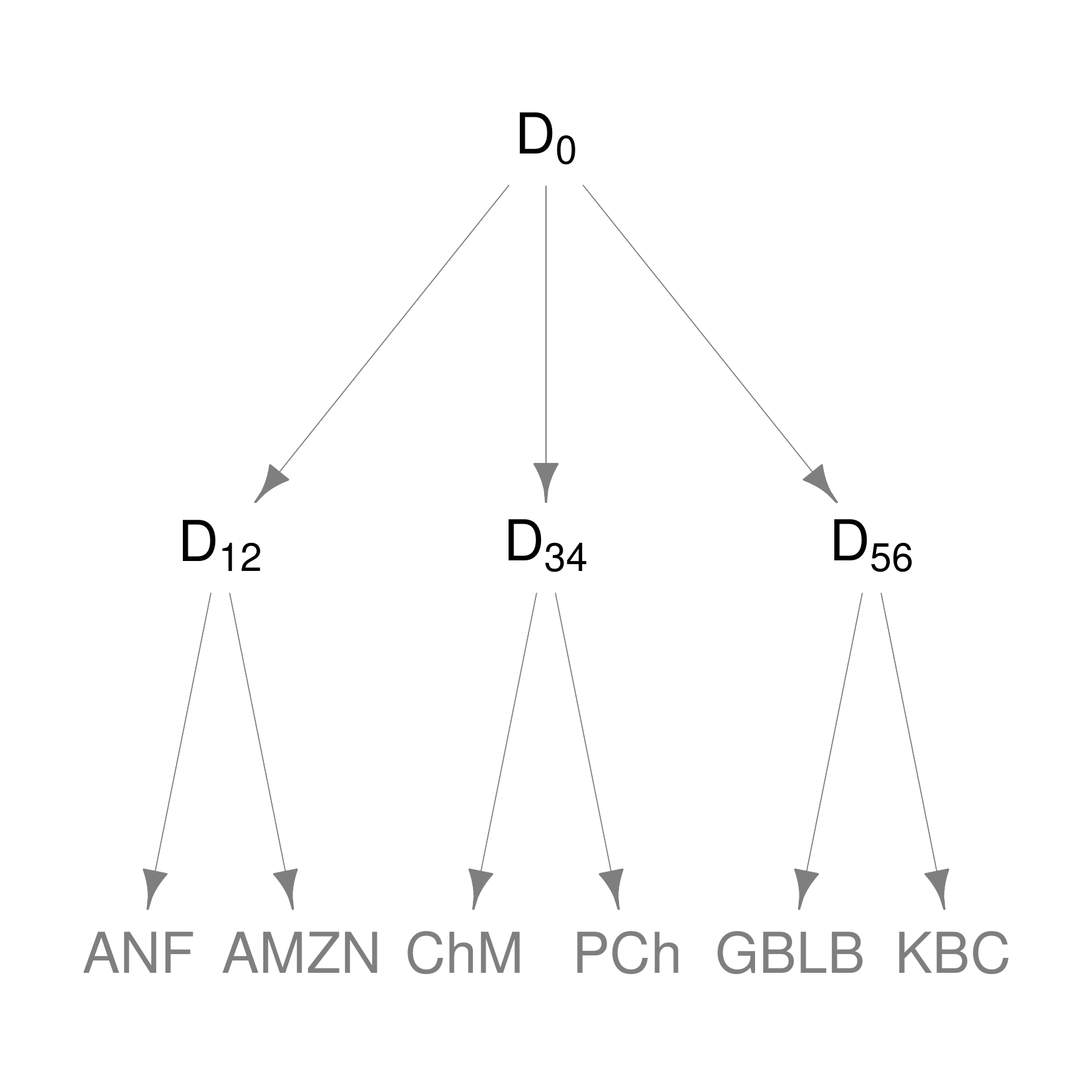}
&
\includegraphics[width=0.49\textwidth]{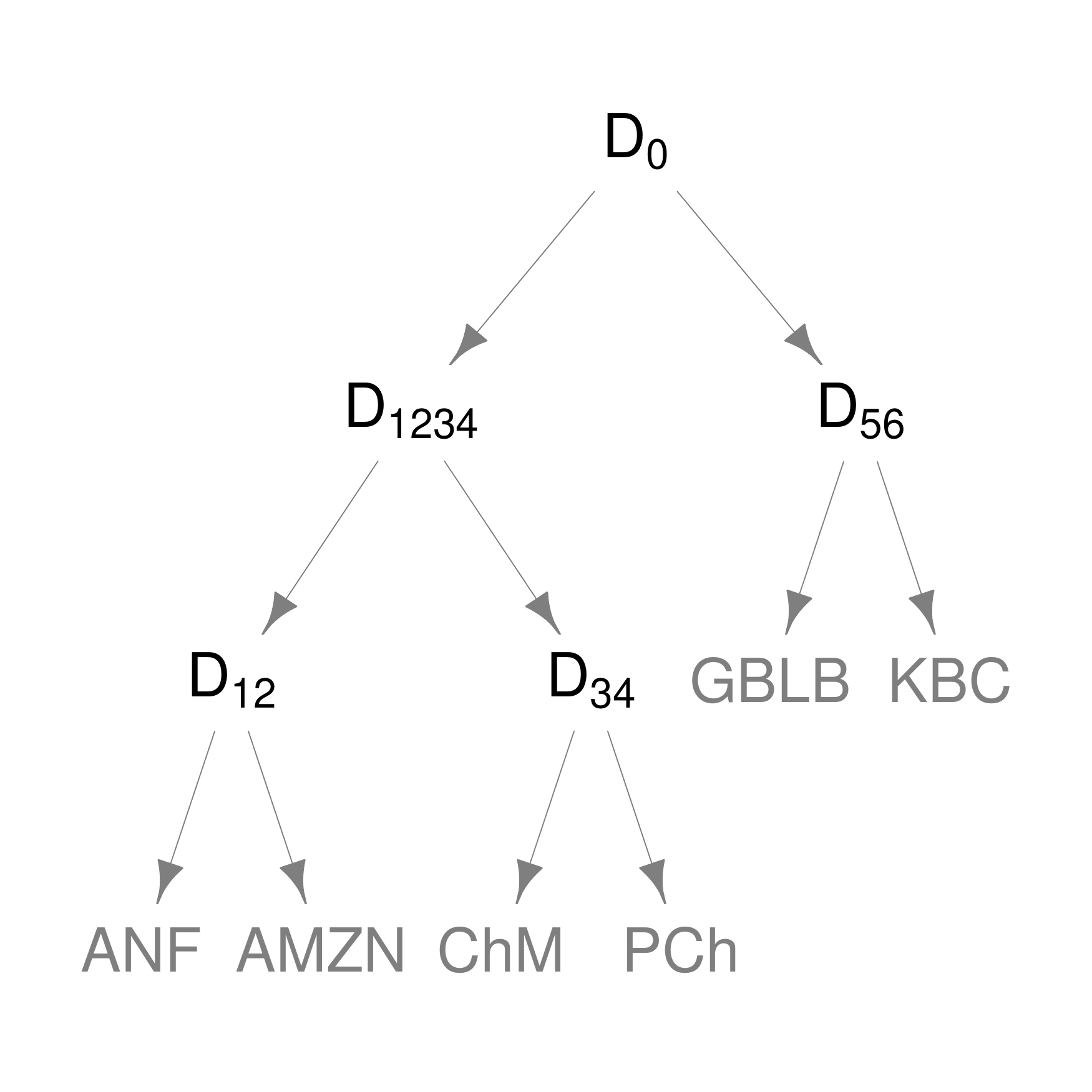}

\end{tabular}
\caption{Figure~11 in \cite{segers2014nonparametric}. \label{wanted}}
\end{figure}

The tree on the right of Figure \ref{wanted} corresponds to the hierarchical NEOFC from Table~\ref{NEOFChierarchy3}, where the symbols $\{\text{D}_0, \text{D}_{1234}, \text{D}_{56}, \text{D}_{12}, \text{D}_{34}\}$ on the right of the table are there to help one compare the table with the tree on the right of Figure \ref{wanted}.

\begin{table}[H]
\centering

\begin{tabular}{rccccccl}
                                  & $i=1$                      & $i=2$                      & $i=3$                      & $i=4$                      & $i=5$                      & $i=6$                      &            \\ \cline{2-7}
\multicolumn{1}{r|}{$j=1$}        & \multicolumn{1}{c|}{$\Pi$} & \multicolumn{1}{c|}{$\Pi$} & \multicolumn{1}{c|}{Frank} & \multicolumn{1}{c|}{Frank} & \multicolumn{1}{c|}{$\Pi$} & \multicolumn{1}{c|}{$\Pi$} & $\text{D}_{34}$   \\ \cline{2-7}
\multicolumn{1}{r|}{$j=2$}        & \multicolumn{1}{c|}{Frank} & \multicolumn{1}{c|}{Frank} & \multicolumn{1}{c|}{$\Pi$} & \multicolumn{1}{c|}{$\Pi$} & \multicolumn{1}{c|}{$\Pi$} & \multicolumn{1}{c|}{$\Pi$} & $\text{D}_{12}$   \\ \cline{2-7}
\multicolumn{1}{r|}{$j=3$}        & \multicolumn{1}{c|}{$\Pi$} & \multicolumn{1}{c|}{$\Pi$} & \multicolumn{1}{c|}{$\Pi$} & \multicolumn{1}{c|}{$\Pi$} & \multicolumn{1}{c|}{Frank} & \multicolumn{1}{c|}{Frank} & $\text{D}_{56}$   \\ \cline{2-7}
\multicolumn{1}{r|}{$j=4$}        & \multicolumn{1}{c|}{Frank} & \multicolumn{1}{c|}{Frank} & \multicolumn{1}{c|}{Frank} & \multicolumn{1}{c|}{Frank} & \multicolumn{1}{c|}{$\Pi$} & \multicolumn{1}{c|}{$\Pi$} & $\text{D}_{1234}$ \\ \cline{2-7}
\multicolumn{1}{r|}{$C^\diamond$} & \multicolumn{6}{c|}{Frank}                                                                                                                                                  & $\text{D}_0$      \\ \cline{2-7}
\end{tabular}
\caption{The hierarchical NEOFC corresponding to the tree structure on the right of Figure \ref{wanted}. \label{NEOFChierarchy3}}

\end{table}

Note that, given the tree structure of a hierarchical Archimedean copula, the corresponding NEOFC will have as many layers as there are internal nodes in the tree when the inner copula is set to independence, or, when the inner copula is used as root node, as many layers as there are internal nodes in the tree minus one.

\begin{example} \label{exhierarchy3}
Set the various parameters of Table \ref{NEOFChierarchy3} to the arbitrary values shown in Table \ref{NEOFChierarchy3_param}. The empirical Kendall's taus based on 10000 observations generated from the related hierarchical NEOFC are shown hereafter.

\begin{code}
> round(cor(generated.data, method="kendall"), 3)
      [,1]  [,2]  [,3]  [,4]  [,5]  [,6]
[1,] 1.000 0.776 0.400 0.406 0.476 0.478
[2,] 0.776 1.000 0.404 0.407 0.475 0.475
[3,] 0.400 0.404 1.000 0.775 0.485 0.483
[4,] 0.406 0.407 0.775 1.000 0.489 0.487
[5,] 0.476 0.475 0.485 0.489 1.000 0.755
[6,] 0.478 0.475 0.483 0.487 0.755 1.000
\end{code}
\end{example}

Note that, in contrast to what is usually required for hierarchical Archimedean copulas, the dependence in hierarchical NEOFCs is not required to be a nondecreasing function as one goes farther away from the root node. The empirical matrix from the last example indeed shows that the dependence at node $D_{12}$ is $\sim0.776$, that of node $D_{1234}=D_{1:4}$ is lower, around $0.400$, but that of the root is higher than $0.400$, at around 0.475.

\begin{table}[H]
\centering

\scriptsize{
\begin{tabular}{rccccccl}
                                  & $i=1$                             & $i=2$                             & $i=3$                             & $i=4$                             & $i=5$                             & $i=6$                             &            \\ \cline{2-7}
\multicolumn{1}{r|}{$j=1$}        & \multicolumn{1}{c|}{$\Pi$}        & \multicolumn{1}{c|}{$\Pi$}        & \multicolumn{1}{c|}{$\tau_{34}=0.4$} & \multicolumn{1}{c|}{$\tau_{34}=0.4$} & \multicolumn{1}{c|}{$\Pi$}        & \multicolumn{1}{c|}{$\Pi$}        & $D_{34}$   \\ \cline{2-7}
\multicolumn{1}{r|}{$j=2$}        & \multicolumn{1}{c|}{$\tau_{12}=0.4$} & \multicolumn{1}{c|}{$\tau_{12}=0.4$} & \multicolumn{1}{c|}{$\Pi$}        & \multicolumn{1}{c|}{$\Pi$}        & \multicolumn{1}{c|}{$\Pi$}        & \multicolumn{1}{c|}{$\Pi$}        & $D_{12}$   \\ \cline{2-7}
\multicolumn{1}{r|}{$j=3$}        & \multicolumn{1}{c|}{$\Pi$}        & \multicolumn{1}{c|}{$\Pi$}        & \multicolumn{1}{c|}{$\Pi$}        & \multicolumn{1}{c|}{$\Pi$}        & \multicolumn{1}{c|}{$\tau_{56}=0.2$} & \multicolumn{1}{c|}{$\tau_{56}=0.2$} & $D_{56}$   \\ \cline{2-7}
\multicolumn{1}{r|}{$j=4$}        & \multicolumn{1}{c|}{$\tau_{1234}=0.1$} & \multicolumn{1}{c|}{$\tau_{1234}=0.1$} & \multicolumn{1}{c|}{$\tau_{1234}=0.1$} & \multicolumn{1}{c|}{$\tau_{1234}=0.1$} & \multicolumn{1}{c|}{$\Pi$}        & \multicolumn{1}{c|}{$\Pi$}        & $D_{1:4}$ \\ \cline{2-7}
\multicolumn{1}{r|}{$C^\diamond$} & \multicolumn{6}{c|}{$\tau_{0}=0.75$}                                                                                                                                                                             & $D_0$      \\ \cline{2-7}
\end{tabular}}

\caption{Parameters for Example \ref{exhierarchy3}, expressed as Kendall's taus for convenience. \label{NEOFChierarchy3_param}}
\end{table}

\end{paragraph}

\begin{paragraph}{Gaussian copulas.}
A Gaussian copula is a copula whose density $c$ satisfies
\begin{align}\label{eq:formula for Gaussian copula}
  \log (c(u_1,\dots,u_d))=
  -\frac{1}{2}\log(\text{det}(R))-\frac{1}{2}z^\top(R^{-1}-I)z,
\end{align}
where $R$ is a $d\times d$ invertible correlation matrix, $z=(z_1,\dots,z_d)^\top$ 
and $z_i$ is the quantile of order $u_i$ of the standard normal distribution.
A Gaussian copula can be represented as in \eqref{eq:representation 2}. 
Let $\Delta=(\Delta_{1},\dots, \Delta_{d})^\top$ be a real vector in $[0,1]^d$ and let $D$ be a diagonal matrix with elements given by $1-\Delta_{i}^2,\,i=1,\dots,d$.
Finally let $C_A$ be a $d$-variate Gaussian copula with correlation matrix $A$.

\begin{proposition}
  \label{prop:GaussianInGaussianOut}
Let $C^\diamond_{Q_0(u_0)}=C^\diamond=C_A$ and let $C_{i}$ be a bivariate Gaussian copula with correlation $\Delta_{i}$, $i \in \{1,\dots,d\}$. Then the outer copula in \eqref{eq:representation 2} is a Gaussian copula with correlation matrix given by $R=D^{1/2}AD^{1/2}+\Delta\Delta^T$.
\end{proposition}

\begin{proof}
Let $R$ be a symmetric nonnegative matrix whose diagonal elements are equal to $1$ and whose element in the $i$-th row and $j$-th column is denoted $\Delta_{ij}$.
Let $(Z_1,\dots,Z_d,Z_0)$ be distributed according to a $(d+1)$-variate centered Gaussian distribution with variance-covariance matrix given by
\begin{align*}
  \begin{pmatrix}
    \underset{d\times d}{R} & \underset{d\times 1}{\Delta} \\
    \underset{1\times d}{\Delta^T}          & \underset{1\times 1}{1}
  \end{pmatrix},  
\end{align*}
so that $(Z_1,\dots,Z_d\vert Z_0=z_0)\sim N(\Delta z_0, \, R-\Delta \Delta^\top)$. The partial correlations are given by
\begin{align*}
  \text{Cov}(Z_i,Z_j\vert Z_0=z_0)=\text{Corr}(Z_i,Z_j\vert Z_0=z_0)=
\frac{\Delta_{ij}-\Delta_{i}\Delta_{j}}
{\sqrt{(1-\Delta_{i}^2)(1-\Delta_{j}^2)}},
\end{align*}
or, in other words, the partial correlation matrix is 
$A=D^{-1/2}(R-\Delta \Delta^\top)D^{-1/2}$.  
Given $Z_0=z_0$, the margins $(Z_i\vert Z_0=z_0)$ are $N(\Delta_{i}z_0,1-\Delta_{i}^2)$, hence 

\begin{equation*}
P(Z_i\leq z\vert Z_0=z_0)=\Phi((z-\Delta_{i} z_0)/\sqrt{1-\Delta_{i}^2}), 
\end{equation*}
and, moreover, the corresponding copula is a Gaussian copula $C_{A}$, with correlation matrix $A$.

Calculate the copula corresponding to $(Z_1,\dots,Z_d)$. Let $\Phi$ be the cumulative distribution function of the univariate standard Gaussian distribution. 
\begin{align*}
  &C_{(Z_1,\dots,Z_d)}(u_1,\dots,u_d)\\
  =&\int P(\Phi(Z_i)\leq u_i,\, i=1,\dots,d\vert Z_0=z_0) \Phi'(z_0)dz_0 \\
  =&\int P(Z_i\leq \Phi^{-1}(u_i),\, i=1,\dots,d\vert Z_0=z_0) \Phi'(z_0)dz_0\\
  =&\int C_{A}\left\{ \Phi\left(\frac{\Phi^{-1}(u_i)-\Delta_{i}z_0}{\sqrt{1-\Delta_{i}^2}}\right), i=1,\dots,d\right\}\Phi'(z_0)dz_0 \\
  =&\int_0^1 C_{A}\left\{ \Phi\left(\frac{\Phi^{-1}(u_i)-\Delta_{i}\Phi^{-1}(u_0)}{\sqrt{1-\Delta_{i}^2}}\right), i=1,\dots,d\right\}du_0.
\end{align*}
This expression corresponds exactly to the copula given in \eqref{eq:representation 2}, with 

\begin{equation*}
\frac{\partial C_i(u_i, u_0)}{\partial u_0}=\Phi\left(\frac{\Phi^{-1}(u_i)-\Delta_{i}\Phi^{-1}(u_0)}{\sqrt{1-\Delta_{i}^2}}\right),
\end{equation*}
see \cite{meyer2013bivariate} for details.

\end{proof}

Note: in the case $d=2$, the outer copula is a Gaussian copula with correlation given by
\begin{equation} \nonumber
\rho_A \sqrt{1-\Delta_1^2} \sqrt{1-\Delta_2^2} + \Delta_1 \Delta_2.
\end{equation}
Also see Example \ref{normalidenti}.

\end{paragraph}

 \begin{paragraph}{C-Vine copulas.}
 Let $(U_0, U_1,\dots,U_d)$ be a random vector following a C-Vine copula distribution truncated after the second level. The density of this truncated C-Vine is given by
 \begin{align} \label{CV}
   c(u_0,\dots,u_d)=\prod_{i=1}^{d-1} c^*_{1,1+i\vert 0}(C^*_{1\vert 0}(u_1\vert u_0),C^*_{1+i\vert 0}(u_{1+i}\vert u_0)\vert u_0) \prod_{i=1}^d c_{i}^{*}(u_i,u_0)
 \end{align} 
where $c^*_{i}=\partial^2 C^*_{i}(u_i, u_0)/\partial u_i \partial u_0$, $C^*_{i\vert 0}(u_i\vert u_0)=\partial C^*_{i}(u_i, u_0)/\partial u_0$, $\{C^*_{i}(u_i, u_0)\}$ is a set of $d$ arbitrary bivariate copulas, and $\{c^*_{1,1+i\vert 0}\}$ is a set of $d-1$ abritrary copula densities for each $u_0$. Due to their extreme flexibility and ease of use (one only has to specify sets of bivariate copulas), Vine copulas have been used in an increasing number of applications and are still a hot topic of research, see for instance \cite{aas_pair-copula_2009}, \cite{kurowicka2011dependence} or  \cite{bedford2002vines}. 
 \begin{proposition} \label{CV_prop}
   If, in \eqref{densityEOFC}, where $T_1=U_0$ and $t_1=u_0$, for each $u_0$, $c_{u_0}$ is defined as
 \begin{align*}
   c_{u_0}(u_1,\dots,u_d)=\prod_{i=1}^{d-1} c^*_{1,1+i\vert 0}(u_1,u_{1+i}\vert u_0),
 \end{align*}
 and $c_{i}(u_i, u_0)=c_{i}^{*}(u_i, u_0)$ for all $i$, then the outer copula $c$ in \eqref{densityEOFC} is the $d$-variate marginal distribution, with respect to $u_0$, of \eqref{CV}, that is, its density is
 \begin{align*}
   c(u_1,\dots,u_d)=&
   \int_0^1
   c_{u_0}(C^*_{1\vert 0}(u_1\vert u_0),\dots,C^*_{d\vert 0}(u_d\vert u_0)) \, 
   \prod_{i=1}^d c^*_{i}(u_i,u_0)
   du_0\\
   =&\int_0^1
   \prod_{i=1}^{d-1} c^*_{1,1+i\vert 0}(C^*_{1\vert 0}(u_1\vert u_0),C^*_{1+i\vert 0}(u_{1+i}\vert u_0)\vert u_0)
   \prod_{i=1}^d c^*_{i}(u_i,u_0) 
   \, du_0.
 \end{align*}
 \end{proposition}

If one assumes that, in \eqref{CV}, none of the elements of $\{c^*_{1,1+i\vert 0}\}$ actually depends on $u_0$, then the inner copula in Proposition \ref{CV_prop} becomes
 \begin{align*}
   c_{u_0}(u_1,\dots,u_d)=\prod_{i=1}^{d-1} c^*_{1,1+i}(u_1,u_{1+i}),
 \end{align*}
which is nothing more than a C-Vine on $(U_1, \ldots, U_d)$, truncated at the first level.

 \end{paragraph}

\begin{paragraph}{$p$-factor models.}
Define respectively $\Pi_1$-factor and $\Pi_2$-factor copulas as copulas of the form
\begin{align}
  &C^{(\Pi_1)}(u_1,\dots,u_d)=\int_0^1 \prod_{i=1}^d C_{i\vert 0}^{(2)}(u_i\vert v_2)\, dv_2, \text{ and}\label{eq:Pi1}\\
  &C^{(\Pi_2)}(u_1,\dots,u_d)=\int_0^1 \int_0^1 \prod_{i=1}^d C_{i\vert 0}^{(2)}(
  C_{i\vert 0}^{(1)}(u_i\vert v_1)\vert v_2)\,dv_2\,dv_1,\label{eq:Pi2}
\end{align}
where $C_{i\vert 0}^{(k)}(u_i\vert v_k)=\partial C_{i}^{(k)}(u_i, v_k)/\partial v_k$ for 
$k=1,2$ and $i=1,\dots,d$, and where the $C_{i}^{(k)}$ are (arbitrary) bivariate copulas. $\Pi_1$-factor and $\Pi_2$-factor copulas have been studied in \cite{krupskii2013factor,krupskiiJoeFactor2015} as copula models for conditionally independent variables given respectively one and two latent factors. 

The following (trivial) proposition aims at recovering $\Pi_1$-factor and $\Pi_2$-factor copulas as special cases of the model \eqref{eq:representation 2}.
\begin{proposition}
  Consider the copulas given in \eqref{eq:Pi1} and \eqref{eq:Pi2}. In \eqref{eq:representation 2}, put $C_{i}=C_{i}^{(2)}$. If, moreover, $C^\diamond_{Q_0(u_0)}=\Pi$ for each $u_0$, then the outer copula $C$ in \eqref{eq:representation 2} is the $\Pi_1$-factor copula given in \eqref{eq:Pi1}. 
  Likewise, if, in \eqref{eq:representation 2}, $C_{i}=C_{i}^{(1)}$ and moreover,
  \begin{align*}
    C^\diamond_{Q_0(u_0)}(u_1,\dots,u_d)=
    \int_0^1 \prod_{i=1}^d
    C_{i\vert 0}^{(2)}(u_i\vert \tilde{u}_0)\,d\tilde{u}_0,
  \end{align*}
then the outer copula $C$ in \eqref{eq:representation 2} is the $\Pi_2$-factor copula given in \eqref{eq:Pi2}.
\end{proposition}

Note that $C^\diamond_{Q_0(u_0)}$ in the above proposition actually does not depend on $u_0$ hence we have conditional invariance. This restriction can be easily removed as follows. Let, for each $u_0$, $\{\widetilde{C}_{i}(\bullet,\bullet;u_0)\}$ be a set of bivariate copulas and 
\begin{align*}
   C^\diamond_{Q_0(u_0)}(u_1,\dots,u_d)=
    \int_0^1 \prod_{i=1}^d
    \widetilde{C}_{i\vert 0}(u_i\vert \tilde{u}_0;u_0)\,d\tilde{u}_0,
\end{align*}
where $\widetilde{C}_{i\vert 0}(u_i\vert \tilde{u}_0;u_0)=\partial \widetilde{C}_{i}(u_i,\tilde{u}_0;u_0)/\partial \tilde{u}_0$. The outer copula is then
\begin{align}\nonumber
  C(u_1,\dots,u_d)=\int_0^1 \int_0^1 \prod_{i=1}^d \widetilde{C}_{i\vert 0}(
  C_{i\vert 0}(u_i\vert u_0)\vert \tilde{u}_0;u_0)\,d\tilde{u}_0\,du_0.  
\end{align}

\end{paragraph}

\section{Inference}
\label{sec:methods}

This section discusses estimation of models of the form \eqref{eq:representation 2} or \eqref{NEOFC}, as well as how one can test for conditional independence for models of the form \eqref{eq:representation 2}.

\subsection{Estimation}
\label{sec:estimation}
Let, in \eqref{eq:representation 2}, $C_{i}(u_i,u_0)=C_{i}(u_i,u_0;\alpha_i)$ and 
\begin{equation*}
C_{x_0}(u_1,\dots,u_d)=C_{x_0}(u_1,\dots,u_d;\beta(x_0)), 
\end{equation*}
where $\beta$ is a mapping which, to each $x_0$ in the support of $X_0$, associates a parameter in the appropriate parameter space. If the mapping $\beta$ depends on a vector of parameters, as in \eqref{eq:correlation in complete example}, we also denote this vector by $\beta$. 
Likewise, we denote the parameter vector which contains the  parameters of the quantile function $Q_0$ of $X_0$ by $\lambda$. 
Accordingly, the notation for the copula of $(U_1,\dots,U_d\vert U_0=u_0)$ becomes $C^\diamond_{Q_0(u_0)}(u_1,\dots,u_d)=C^\diamond(u_1,\dots,u_d;\beta,\lambda)$.
Finally, let us denote by $(x_{h1},\dots,x_{hd})$, $h=1,\dots,n$, the sample of the distribution $F$ with margins $F_1,\dots,F_d$ and copula $C$.

The pseudo log likelihood function to maximize is
\begin{multline}   \label{eq:likelihood}
 L_n(\boldsymbol{\theta})=\sum_{h=1}^n  \log  \int_0^1 c\left[
    C_{1\vert 0}(\widehat{F}_1(x_{h1})\vert u_0;\alpha_1),\dots,
    C_{d\vert 0}(\widehat{F}_d(x_{hd})\vert u_0;\alpha_d)
    ;\beta,\lambda  \right] \\
	\times \prod_{i=1}^d
  c_{i}(\widehat{F}_i(x_{hi}), u_0;\alpha_i)
  \,du_0,
\end{multline}
where $\boldsymbol{\theta}$ stands for the complete parameter vector, that is, $\boldsymbol{\theta}=(\alpha,\beta,\lambda)$, $\alpha=(\alpha_1,\dots,\alpha_d)$ and 
$\widehat{F}_i$ denotes an estimate of $F_i$, $i\in\{1,\dots,d\}$.
There are many ways to estimate $F_i$. For instance, $\widehat{F}_i$ may be the empirical distribution function, as in  \cite{genest1995semiparametric}, or may be a parametric estimate, as in \cite{joe1996estimation}.

Regarding the computational aspects, especially in higher dimensions and for large datasets, the likelihood \eqref{eq:likelihood} may be costly to compute due to the repeated use of integrals (as many as the sample size). A brief discussion on these computational aspects are given in Section \ref{practical_issues}.

\subsection{Testing for conditional independence \label{test}}
This section provides procedures to test for conditional independence in models based on the representation \eqref{eq:representation 2}. Indeed, being able to assess  if the variables of interest are dependent or independent conditioned on the latent factor seems a crucial issue. Conditional independence would mean that the factor captures all the dependence in the data whereas no conditional independence would mean that there is a remaining, intrinsic dependence in the variables even though the factor has been accounted for. 

Throughout this subsection, the bivariate copulas $C_{i}$, $i \in \{1, \dots, d\}$, are assumed to belong to some known parametric families. The inner copula $C_{u_0}=C^\diamond_{Q_0(u_0)}$, however, can be left fully unspecified: either nothing is known about it, either a parametric form is assumed. The possibility to carry out a hypothesis test in this setting is new to the literature.

The hypothesis test for conditional independence is of the form 
\begin{align*}
&H_0:\,C^\diamond_{Q_0(u_0)}=\Pi \text{ for all }u_0\notag\\
\text{ versus }
&H_1:\text{ there exists some }u_0\text{ such that }C^\diamond_{Q_0(u_0)}\ne\Pi
\end{align*}
(recall that $\Pi$ stands for the independence copula or the product function), where for two functions $f$ and $g$, $f=g$ means that $f(t)=g(t)$ for all $t$ in their domain. 

If a certain parametric form is assumed for $C^\diamond_{Q_0(u_0)}$, such as in Section \ref{sec:spectralRepresentation}, then most likely the test will reduce to testing for a parameter to equate a certain value, and no conceptual difficulties refrain the task. For instance, in \eqref{eq:correlation in complete example}, testing for conditional independence amounts to testing for $\beta_0=\infty$ or $\beta_1=\infty$ (conceptually). Let us remark that testing for conditional invariance is also feasible: one need to test $\beta_1=0$. 

If $C^\diamond_{Q_0(u_0)}$ is left fully unspecified, the alternative hypothesis needs to be slightly restricted in order for a test to exist. Consider
\begin{align}\label{eq:hypothesis test}
&H_0:\, C^\diamond_{Q_0(u_0)}=\Pi \text{ for all }u_0\notag\\
\text{ versus }
&H_1:\, C^\diamond_{Q_0(u_0)}>\Pi\text{ for all } u_0.
\end{align}
Then, the alternative hypothesis is: ``conditioned on the factor, the variables of interest are positively dependent''. 

Let $\pi$ be the risk of type I error. One rejects $H_0$ if $T_n \leq c_{\pi}$, where $c_{\pi}$ is chosen so that $P_{H_0}(T_n\leq c_{\pi})=\pi$ and where
\begin{align} \label{stat_test}
  T_n=\underset{t\in[0,1]^d}{\sup} (M(t)-\widehat{C}(t)),
\end{align}
where $M(t)=M(t_1,\dots,t_d)=\min(t_1,\dots,t_d)$ is the Fr\'echet-Hoeffding upper bound for copula and 
\begin{align}\label{eq:empirical copula}
  \widehat{C}(t)=\frac 1 n \sum_{h=1}^n\mathbf{1}(\{\widehat{F}_i(X_{hi})\leq t_i\},\,i=1,\dots,d)
\end{align}
is the empirical estimator of $C$ ; $(X_{h1},\dots,X_{hd})$, $h \in \{1,\dots,n\}$, being the data ; and $\widehat F_i$ being the empirical distribution function of $X_{hi}$.

The heuristic underlying the expression of $T_n$, the test statistic, is as follows. 
Denote by $C^{(H_0)}$ the outer copula under $H_0$, that is, one substitutes $\Pi$ for the inner copula $C^\diamond_{Q_0(u_0)}$ in \eqref{eq:representation 2} and gets
\begin{align}\label{eq:outer copula under H0}
  C^{(H_0)}(u_1,\dots,u_d)=\int_0^1 \prod_{i=1}^d C_{i\vert 0}(u_i\vert u_0)\, du_0.
\end{align}
If $H_0$ is true, then the true copula $C$ verifies $C=C^{(H_0)}$. But if $H_0$ is false, $C_{u_0}>\Pi$ implies $C>C^{(H_0)}$. In order to reject the null, one could therefore compare the empirical copula $\widehat{C}$ to $\widehat{C}^{(H_0)}$, the latter being an estimate of $C^{(H_0)}$, and reject the null if $\widehat{C}-\widehat{C}^{(H_0)}$ is too large. However, as the distribution of this difference under $H_0$ is unknown, one would need to use bootstrap and fit $C^{(H_0)}$ to simulated data a large number of times, a process far too expensive for most hardware.

To overcome this issue, the comparison of $\widehat{C}$ to $\widehat{C}^{(H_0)}$ is done indirectly using the Fr\'echet-Hoeffding upper bound as a baseline (note that the Fr\'echet-Hoeffding lower bound for copulas could also be considered as a baseline for comparison, the upper one was however picked since it has the nice property to be a copula $\forall d$, although this property is by no means required for the test). The result of this suggestion is \eqref{stat_test}, which describe a test statistic always positive and expected to take small values when $H_0$ is false.

The bootstrap of $T_n$ is now described. Note that under $H_0$, the outer copula in \eqref{eq:outer copula under H0} is fully parametric: one can obtain an estimate $\widehat C^{(H_0)}$ by maximum pseudo-likelihood \citep{krupskii2013factor}. New bootstrap samples (say $N$ of them) are then drawn using Algorithm \ref{data_gen} in order to get $N$ test statistics $T_n^{(1)},\dots,T_n^{(N)}$. These can be used, for instance, to compute a p-value as $P_{H_0}\bigg(T_n\leq T_n^{(obs)}\bigg)\approx N^{-1}\sum_{k=1}^N \mathbf{1}\bigg(T_n^{(k)}\leq T_n^{(obs)}\bigg)$. In other words, the p-value is estimated by counting how many test statistics in $\{T_n^{(t)}\}, t \in \{1, \ldots, N\}$, are lower or equal than the test statistic $T_n^{(obs)}$, built using the data from the original sample.

%Comparing the nonparametric estimator $\widehat C$ in \eqref{eq:empirical copula} to a parametric estimator under $H_0$, say $\widehat C_{\text{parametric}}^{(H_0)}$, for instance by considering Kolmogorov-Smirnov or Cram\'er-von Mises distances, would have been possible but would have required, because of the bootstrap procedure, the computation of $\widehat C_{\text{parametric}}^{(H_0)}$ as many times as they are bootstrap samples, which increases the computational needs.

Finally, let us note that the test $H_0: C^\diamond_{Q_0(u_0)}=\Pi$ against $H_1: C^\diamond_{Q_0(u_0)}< \Pi$ can be carried out by considering \eqref{stat_test} again, but this time with a rejection region on the right, that is, we reject $H_0$ if $T_n\ge c_\pi$, where $c_{\pi}$ is chosen so that $P_{H_0}(T_n\geq c_{\pi})=\pi$.

\section{Illustrations}
\label{sec:numericalExperiments}

The purpose of this section is to illustrate how one can take advantage of the framework presented in Section \ref{sec:spectralRepresentation} in practice and to study the power of the test statistic $T_n$ in \eqref{stat_test} by means of a simulation experiment. We first provide a few technical details on how some numerical operations were performed.

\subsection{Computational aspects} \label{practical_issues}
In this section, log-likelihoods are maximized using either gradient descent algorithms, which can be found in the \verb+optim+ function of the statistical software \verb+R+, or differential evolution \citep{storn1997differential, qin2005self, fregrgdhgth}, for which a \verb+R+ package also exists: \verb|DEoptim| \citep{refrfre}.

Note that gradient descent algorithms usually require to provide a starting parameter vector. It is advised to try several such points and retain the best result for the likelihood. 

For the numerical evaluation of the integral in \eqref{eq:likelihood}, either naive Monte-Carlo integration was used, where the integration space is sampled a thousand times according to a uniform distribution that stretches over the integration space, or adaptative multidimensional integration was used \citep{cubaR}.

Note that in case of a one-factor copula or of an extended one-factor copula, the integration space is the segment $[0,1]$. However, with a nested extended one-factor copula, the integration space becomes $[0,1]^w$ and is therefore multi-dimensional.

\subsection{Revisiting the daily log returns from \cite{segers2014nonparametric}}

In \cite{segers2014nonparametric}, raw daily log returns from January 2010 to December 2012 of the following indices were gathered with the help of Yahoo! Finance: 
\begin{itemize}
  \setlength{\itemsep}{1pt}
  \setlength{\parskip}{0pt}
  \setlength{\parsep}{0pt}
\item Abercrombie \& Fitch Co. (ANF), traded in New York;
\item Amazon.com Inc. (AMZN), traded in New York;
\item China Mobile Limited (ChM), traded in Hong Kong;
\item PetroChina (PCh), traded in Hong Kong;
\item Groupe Bruxelles Lambert (GBLB), traded in Brussels;
\item and KBC Group (KBC), traded in Brussels.
\end{itemize}

They then analyzed the filtered data using hierarchical Archimedean copulas. The result was the estimation of the two tree structures in Figure \ref{wanted}. Refer to their paper for more details.

In this subsection, we fit the hierarchical NEOFC described by Table \ref{NEOFChierarchy3} on the same filtered data, in an effort to reproduce their results. The fit is the result of the maximization of the related likelihood using the \verb|DEoptim| function. The estimated parameters for the model in Table~\ref{NEOFChierarchy3} are, expressed as Kendall's taus for convenience,

\begin{itemize}
  \setlength{\itemsep}{1pt}
  \setlength{\parskip}{0pt}
  \setlength{\parsep}{0pt}
\item $\tau_0=0.176813$,
\item $\tau_{1234}=0.072114$,
\item $\tau_{12}=0.251929$,
\item $\tau_{34}=0.354454$,
\item and $\tau_{56}=0.383513$.
\end{itemize}
Also refer to Table \ref{NEOFChierarchy3_param}.

The Kendall's taus between the random variables of interest induced by the fitted model are however unknown but can be calculated through simulation. Hereafter is the matrix of Kendall's taus between the random variables of interest, calculated using 10000 observations from the fitted hierarchical NEOFC.

\begin{code}
      ANF AMZN  ChM  PCh GBLB  KBC
 ANF 1.00 0.28 0.13 0.14 0.12 0.12
AMZN .... 1.00 0.12 0.12 0.13 0.13
 ChM .... .... 1.00 0.34 0.10 0.11
 PCh .... .... .... 1.00 0.10 0.11
GBLB .... .... .... .... 1.00 0.36
 KBC .... .... .... .... .... 1.00
\end{code}

This can be compared to the matrix of Kendall's taus induced by the tree structure on the right of Figure \ref{wanted}, also see \cite{segers2014nonparametric}:

\begin{code}
      ANF AMZN  ChM  PCh GBLB  KBC
 ANF 1.00 0.31 0.08 0.08 0.18 0.18
AMZN .... 1.00 0.08 0.08 0.18 0.18
 ChM .... .... 1.00 0.35 0.18 0.18
 PCh .... .... .... 1.00 0.18 0.18
GBLB .... .... .... .... 1.00 0.44
 KBC .... .... .... .... .... 1.00
\end{code}

The Kendall's taus from the fitted hierarchical NEOFC however match the ones induced by the model on the left of Figure \ref{wanted} more:

\begin{code}
      ANF AMZN  ChM  PCh GBLB  KBC
 ANF 1.00 0.31 0.14 0.14 0.14 0.14
AMZN .... 1.00 0.14 0.14 0.14 0.14
 ChM .... .... 1.00 0.35 0.14 0.14
 PCh .... .... .... 1.00 0.14 0.14
GBLB .... .... .... .... 1.00 0.44
 KBC .... .... .... .... .... 1.00
\end{code}

In conclusion, the hierarchical NEOFC described by Table \ref{NEOFChierarchy3}, fitted on the filtered data from \cite{segers2014nonparametric}, is able to capture a similar hierarchical dependence structure as the one captured by the tree on the left of Figure \ref{wanted}.

\subsection{Power of the test for conditional independence}\label{sec:illustrations testing}
In this section, we study the power of the test statistic $T_n$ in \eqref{stat_test} by means of a simulation experiment. 
%Recall that the power is the probability of rejecting the null hypothesis $H_0$ under the alternative hypothesis $H_1$.
We considered the test \eqref{eq:hypothesis test}
and set the type I error risk to $\pi=0.1$. We drew $N=500$ datasets of size $n=50$ and $500$ from 
the model \eqref{eq:representation 2}, with $d=3$ and $C_{i}$ being Clayton copulas as in \eqref{eq:formula for Clayton bivariate copula}
with parameters $\alpha_i$, $i=1,2,3$, such that Kendall's $\tau$ coefficients are equal to 0.4 for $i=1$, 0.5 for $i=2$ and 0.6 for $i=3$. The inner copula $C_{x_0}$ was a normal copula as in \eqref{eq:formula for Gaussian copula} with correlation matrix 
\begin{align}\label{eq:correlation matrix x_0}
  R=
  \begin{pmatrix}
    1& & & \\
     &\ddots&\beta& \\
     &\beta&\ddots& \\
     & & &1
  \end{pmatrix},
\end{align}
for $\beta=0.0, 0.1, 0.2, 0.3, 0.4$, and, finally, $0.5$. (There are $N=500$ samples of size $n=50$ and $500$ for each $\beta$). Note that $\beta=0$ corresponds to the null hypothesis $H_0$.

For each sample, $k=1,\dots,N$, we calculated a $p$-value $p^{(k)}$ based on 200 boostrap replications. That is, we calculated the proportion of 200 simulated test statistics that where lower or equal than the observed one. As rejection occurs whenever the $p$-value is lower or equal to the type I error risk $\pi$, we approximated the power by the proportion of the $p^{(k)}$  falling below $\pi$. See Section \ref{test} for details.

Figure \ref{power_study} shows the estimated power of $T_n$ in \eqref{stat_test}. As it was expected, the power of the test increases as $n$ and $\beta$ grow. Furthermore the power is equal to the type I error risk $\pi$ under the null, that is when $\beta=0$.

\begin{figure}[H]
\centering
\begin{tabular}{c}

\includegraphics[width=0.5\textwidth]{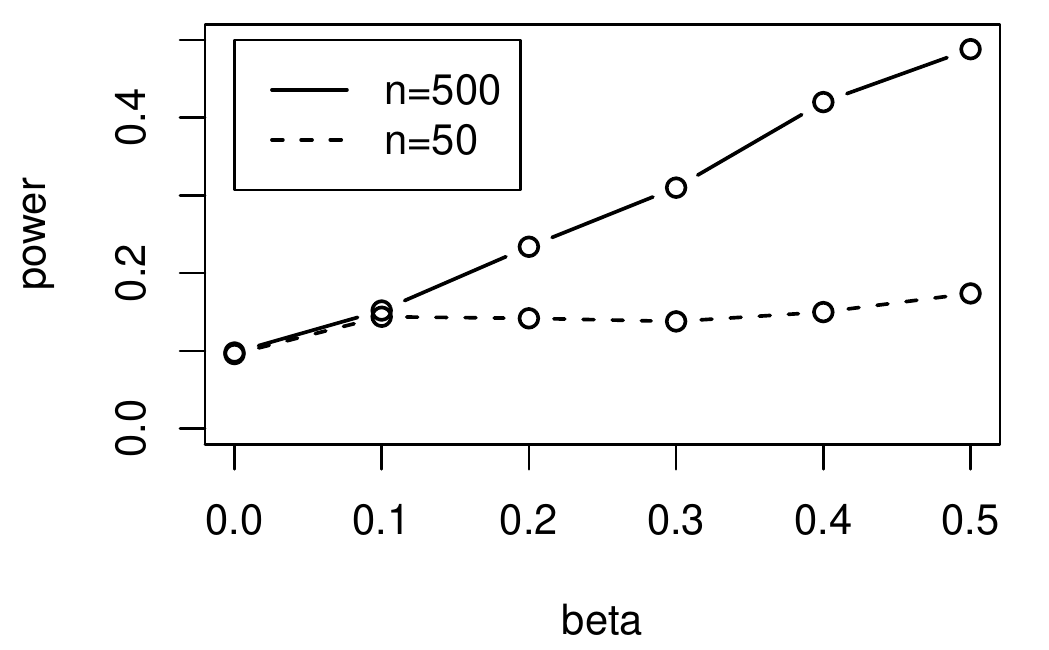}
\end{tabular}
\caption{Power of \eqref{stat_test} as a function of $\beta$, for the setting described above. \label{power_study}}
\end{figure}

\section{Conclusion}
In this paper, we extended the scope of one-factor copulas thanks to Equations \eqref{eq:representation 1}, \eqref{eq:representation 2} and \eqref{NEOFC}. Models built using these equations can now feature a varying conditional dependence structure and a factor's distribution not restricted to be the standard uniform. General dependence properties for these models were given and explicit examples were discussed. The usefulness of these models was illustrated by considering the estimation of the dependence of 6 daily log returns already analyzed in \cite{segers2014nonparametric}. Furthermore, a novel hypothesis test was constructed in order to assess whether conditional independence holds or not, and the power of that test was studied through a simulation experiment.

\newpage
\phantomsection
%\addcontentsline{toc}{chapter}{Bibliography} % without this, no Biography in the toc.
\bibliographystyle{plainnat}

\end{document}